\documentclass[10pt,journal]{IEEEtran}

\usepackage[T1]{fontenc}
\usepackage[utf8]{luainputenc}
\usepackage{float}
\usepackage{amsmath}
\usepackage{amssymb}
\usepackage{graphicx}
\usepackage{esint}
\usepackage{cite}
\usepackage{color}
\usepackage{xcolor}
\usepackage{subfig}
\usepackage{epstopdf}
\newtheorem{prop}{Proposition}

\begin{document}

\title{Diffusion Maps Kalman Filter for a Class of Systems with Gradient Flows}

\author{Tal~Shnitzer,~\IEEEmembership{Student Member,~IEEE}, Ronen~Talmon,~\IEEEmembership{Member,~IEEE}, Jean-Jacques~Slotine
\thanks{Tal~Shnitzer and Ronen~Talmon are with the Department of Electrical Engineering, Technion -- Israel Institute of Technology, Technion City, Haifa, Israel 32000 (e-mail: shnitzer@campus.technion.ac.il; ronen@ee.technion.ac.il). 
}
\thanks{Jean-Jacques E. Slotine is with the Nonlinear Systems Laboratory, Massachusetts Institute of Technology, Cambridge, Massachusetts 02139, U.S.A. (e-mail: jjs@mit.edu).}}
\maketitle

\begin{abstract}
In this paper, we propose a non-parametric method for state estimation of high-dimensional nonlinear stochastic dynamical systems, which evolve according to gradient flows with isotropic diffusion. We combine diffusion maps, a manifold learning technique, with a linear Kalman filter and with concepts from Koopman operator theory. 
More concretely, using diffusion maps, we construct data-driven virtual state coordinates, which linearize the system model. Based on these coordinates, we devise a data-driven framework for state estimation using the Kalman filter.
We demonstrate the strengths of our method with respect to both parametric and non-parametric algorithms in three tracking problems. 
In particular, applying the approach to actual recordings of hippocampal neural activity in rodents directly yields a representation of the position of the animals.
We show that the proposed method outperforms competing non-parametric algorithms in the examined stochastic problem formulations. Additionally, we obtain results comparable to classical parametric algorithms, which, in contrast to our method, are equipped with model knowledge.
\end{abstract}


\section{Introduction}

\IEEEPARstart{I}{n} many real applications, the system model is not accessible and some estimation of it is required. 
State estimation and characterization of stochastic, possibly nonlinear, dynamical systems are therefore widely studied problems. Traditionally, such problems are addressed using classical algorithms, which rely on predefined parametric models. On the one hand, parametric models need to be sufficiently simple to allow accurate parameter estimation from measurements. On the other hand, too simple models often fail to accommodate the complexity of real systems.
This facilitates the development of non-parametric methods \cite{deisenroth2012robust,wang2014bayesian,surana2016linear,brunton2017HAVOK,berry2015nonparametric}.
Particularly in this paper, we take a non-parametric approach and propose a new method to derive the system model in a data-driven manner. 

To demonstrate the primary idea, consider a classical nonlinear Simultaneous Localization and Mapping (SLAM) problem, where the goal is to track the 2D position $\boldsymbol{x}=(x_1,x_2)^T$ of a moving object. Typically, in the bearing-only version of this problem, the accessible system measurements are given by the azimuth of an object, i.e. by $\phi=\mathrm{arc}\tan\left(x_1/x_2\right)$. This nonlinear model, relating the state (position) of the system $\boldsymbol{x}$ to the measurements $\phi$, poses a challenge for processing and analysis, since traditional linear methods cannot be applied.
In \cite{tan2018SLAMl}, this nonlinearity is resolved by constructing virtual measurements, $y$, and measurement mapping, $\boldsymbol{H}$, based on the knowledge of the system properties and model, which allow for the formulation of an equivalent linear problem and the application of a linear time-varying Kalman filter \cite{Kalman1960}. Briefly, since the system state is given by $\boldsymbol{x}=(\sin\phi,\cos\phi)$, linearization is achieved by defining the virtual measurements $y=\boldsymbol{H}\boldsymbol{x}+v$, where $v$ is noise and:
\begin{equation}
\boldsymbol{H} = 
\begin{bmatrix}
\cos\phi & -\sin\phi
\end{bmatrix}
\end{equation}
such that $\boldsymbol{H}\boldsymbol{x}=0$.
The resulting measurement equation is linear and can be constructed from the given measurements $\phi$.
In this work, analogously to the virtual measurements that linearize the system, 
we  propose a computational method to construct the data-driven non-parametric counterparts of a virtual state, linearizing the system dynamics and measurement model. However, in contrast to \cite{tan2018SLAMl}, our construction of this virtual state is data-driven and does not require any explicit knowledge of the system properties or measurement modality, e.g. the knowledge that the measurements represent the azimuth in a 2D tracking problem.

One way to devise such a computational method, which has recently drawn significant research attention, is to address the problem of data-driven system analysis and state estimation from an operator-theoretic point of view. In this approach, the dynamical system is described by two dual operators, the Perron-Frobenius operator, which represents probability density evolution, and the Koopman Operator, which is defined on some linear functional space of infinite dimension and describes the time-evolution of observables \cite{mezic2005spectral,budivsic2012applied}. 

In the context of empirical dynamical systems analysis, the main challenge is to approximate these operators from the system measurements.
Several methods for estimating the Koopman Operator have been proposed in recent years \cite{Williams2015edmd,budivsic2012applied,mezic2004comparison,mohr2014construction,ulam1960problems,klus2015numerical}. 
For example, the Extended Dynamic Mode Decomposition (EDMD) \cite{Williams2015edmd,budivsic2012applied} approximates the Koopman eigenfunctions and modes based on two sets of points, related through system dynamics, and a set of dictionary elements. However, the optimal choice of the dictionary in EDMD depends on the data \cite{Williams2015edmd}. 
This framework for estimating the Koopman eigenfunctions and modes was later employed in \cite{surana2016linear} as part of a non-parametric Kalman filtering framework, where a linear Kalman filter is constructed based on the approximation of the Koopman Operator, along with its eigenvectors, eigenvalues and Koopman modes using EDMD. The Kalman filter propagates the system in the space spanned by the Koopman eigenvectors and then the resulting estimates are projected back to the state space.
Other work on the analysis of nonlinear stochastic dynamical systems based on Koopman operator theory includes \cite{mezic2017koopmanRandom,mezic2004comparison,mohr2014construction}. The work in \cite{mezic2017koopmanRandom} offers a formal definition and rigorous mathematical analysis for the generalization of the Koopman operator to stochastic dynamical systems. In addition, a new framework for approximating the eigenfunctions and eigenvalues of this stochastic Koopman operator is presented.
A different approach is proposed in \cite{mezic2004comparison} and \cite{mohr2014construction}, where the authors characterize the long term behavior of a system (asymptotic dynamics) based on time-averages of functions. In \cite{mezic2004comparison}, invariant-measures are defined and it is shown that these measures coincide with the eigenfunctions of the Koopman Operator and can be simply calculated by the Fourier transform. In \cite{mohr2014construction}, this framework is extended to address dynamical systems which are not measure-preserving using Laplace averages. In both \cite{mezic2004comparison} and \cite{mohr2014construction}, several trajectories of the system are needed for the analysis. 

Different related non-parametric frameworks for state estimation in stochastic dynamical systems based on geometry and manifold learning were presented in \cite{berry2014nonparametric,berry2015nonparametric,hamilton2016ensemble,giannakis2015data}. In \cite{berry2015nonparametric} and \cite{berry2014nonparametric} a probabilistic approach is proposed, in which the problem is projected onto coordinates constructed by diffusion maps \cite{Coifman2006}, a manifold learning technique. In these diffusion maps coordinates, the probability density function of the system state can be propagated in time without prior knowledge of the system dynamics, yet, the underlying system state is assumed to be accessible.
The work in \cite{hamilton2016ensemble} proposes to construct an ensamble Kalman filter based on delay embedding coordinates (Takens embedding \cite{takens1981dynamical}), used for dynamics estimation. The propagation in time is estimated based on nearest neighbors of the current time-lag frame.
In \cite{giannakis2015data}, a framework for estimating the eigenfunctions of the Koopman Operator generator based on diffusion maps is presented. This work discusses the relationship between diffusion maps and the Koopman operator and proposes to estimate the Koopman generator eigenfunctions using the eigenfunctions of the Laplace-Beltrami operator approximated by diffusion maps.

The relationship between diffusion maps and the Koopman Operator is further discussed in \cite{mezic2005spectral}. There, it is shown that EDMD can also provide data-driven dimensionality reduction. Moreover, for systems in which the system state is described by a Markov process, the eigenfunctions of the backward Fokker-Planck operator can be approximated using EDMD, similarly to diffusion maps. The main benefit of such a manifold learning approach using EDMD is that both the dynamics and the geometry of the underlying state are taken into account. 

In this paper, we present a framework for state estimation of stochastic dynamical systems, with state dynamics of gradient flows and isotropic diffusion, which reveals the system model with minimal prior assumptions, using diffusion maps \cite{Coifman2006} and the Kalman filter. 
We assume that we are given a set of noisy measurements from some unknown nonlinear function of a stochastic underlying state and show that a \emph{linear} model describing the system can be revealed, even for highly nonlinear systems. This is obtained in a completely data-driven manner, based on virtual state coordinates constructed by diffusion maps and their inherent dynamics \cite{coifman2008diffusion}.
A Kalman filter is then formulated based on the recovered system model and utilized for non-parametric state estimation. By constructing the Kalman filter based on the recovered model, we incorporate system dynamics into diffusion maps, combining geometry and dynamics, as in \cite{mezic2005spectral}, from a new manifold learning standpoint.
We further show that our method uncovers an operator describing the system dynamics, which is analogous to the Stochastic Koopman Operator \cite{mezic2005spectral}, an extension of the Koopman Operator for stochastic systems. Moreover, our devised method is well suited for high-dimensional systems due to the nonlinear dimensionality reduction obtained by diffusion maps.

The framework in this paper is tightly related to \cite{shnitzer2016}, in which a linear observer was constructed by exploiting similar properties of the diffusion maps coordinates and their dynamics. We show that the Kalman filter framework introduced in the present work outperforms the observer framework both by reducing the amount of hyper-parameters and by significantly improving the results in certain scenarios.

With respect to previous work, our method encompasses several key differences. First, it does not rely on accessibility to the state of the system (as assumed in \cite{berry2014nonparametric,berry2015nonparametric}).
Second, it does not require predefined dictionary elements (as required in \cite{Williams2015edmd}). Third, our method addresses stochastic dynamical systems in contrast to the deterministic systems considered in \cite{hamilton2016ensemble,surana2016linear}. Finally, it does not require a training set of samples with known states, in contrast to common non-parametric filtering algorithms \cite{deisenroth2012robust,wang2014bayesian}.

The remainder of the paper is organized as follows. In Section \ref{sec:Formulation}, the general problem setting is presented. In Section \ref{sec:Theory}, we first overview the key-points of our method and then describe the derivation of the Diffusion Maps Kalman filter framework in detail. In Section \ref{sec:Results}, we demonstrate the strengths of our method on two object tracking problems. We compare our method with both parametric and non-parametric methods. Section \ref{sec:Conclusion} concludes the paper with a short summary.

\section{Problem Formulation\label{sec:Formulation}}

Consider an ergodic stochastic dynamical system with some nonlinear generator $S:\mathcal{M}\rightarrow \mathcal{M}$, where $\mathcal{M}$ is a compact Riemannian manifold of dimension $d$. The system is defined by:
\begin{eqnarray}
\dot{\boldsymbol{\theta}}_{t} & = & S\left(\boldsymbol{\theta}_{t},\boldsymbol{u}_{t}\right)\label{eq:DynGen}\\
\boldsymbol{z}_{t} & = & g\left(\boldsymbol{\theta}_{t}\right)+\boldsymbol{v}_{t}\label{eq:MeasEq}
\end{eqnarray}
where $\boldsymbol{\theta}_t\in\mathcal{M}$ is the system state, $\dot{\boldsymbol{\theta}}_t$ is its time derivative, $\boldsymbol{u}_t\in\mathbb{R}^d$ is the process noise, $\boldsymbol{z}_t\in\mathbb{R}^m$ are the system measurements through some unknown nonlinear function $g$ and $\boldsymbol{v}_t\in\mathbb{R}^m$ is the measurement noise.
The evolution in time of such a dynamical system can be described by the Stochastic Koopman Operator \cite{mezic2005spectral}, which is defined by
\begin{equation}
\left(U_{st}f\right)\left(\boldsymbol{\theta}_t\right)=\mathbb{E}\left[f\circ T\left(\boldsymbol{\theta}_t,\boldsymbol{u}_t\right)\right]
\end{equation}
where $T$ is the flow induced by $S$, $\circ$ is the composition operator, $f:\mathcal{M}\rightarrow\mathbb{R}$ are some observables from an infinite dimensional vector space, closed under composition with $T$, and $\mathbb{E}$ denotes expectation.

One of the notable properties of the Koopman Operator, which has increased its usage in a line of recent work, is that it is linear in the space of observables, even for highly nonlinear dynamical systems: $U_{st}\left(\alpha_1 f_1+\alpha_2 f_2\right)=\mathbb{E}\left[\left(\alpha_1 f_1+\alpha_2 f_2\right)\circ T\right]=\alpha_1\mathbb{E}\left[f_1\circ T\right]+\alpha_2\mathbb{E}\left[f_2\circ T\right]=\alpha_1 U_{st}f_1+\alpha_2 U_{st}f_2$.
However, the tradeoff is that even for finite dimensional dynamical systems, the Koopman Operator acts on an infinite dimensional space of observables.

In this work, we focus on the case in which the dynamics equation \eqref{eq:DynGen} takes the form of a Langevin equation: 
\begin{equation}
\dot{\boldsymbol{\theta}}_{t}=-\nabla U\left(\boldsymbol{\theta}_t\right) + \sqrt{\frac{2}{\beta}}\dot{\boldsymbol{u}}_t\label{eq:Langevin}
\end{equation}
where $U\left(\boldsymbol{\theta}_t\right)$ is a smooth and bounded potential, $\sqrt{\frac{2}{\beta}}$ is a constant diffusion coefficient, $\boldsymbol{u}_t$ is Brownian motion and $\dot{\boldsymbol{u}}_t$ is its time derivative. 

The goal in this work is to build a new coordinate system representing the hidden system state $\boldsymbol{\theta}_t$ and devise a filtering framework in these new coordinates, given only the measurements $\boldsymbol{z}_t$, without any prior knowledge on the system equations \eqref{eq:DynGen} and \eqref{eq:MeasEq}.

\section{Diffusion-Based Kalman Filter\label{sec:Theory}}
In this section we lay the foundation for our proposed framework and present the theoretical results which allow for the construction of a data-driven linear Kalman filter.

\subsection{Overview\label{sub:Overview}}
We present a method based on diffusion maps that discovers a new coordinate system describing a model of the state of the system with \emph{linear} drift in a completely data-driven manner. 
Based on this linear drift, our method constructs a linear operator, analogous to the Stochastic Koopman Operator, from measurements, without prior model knowledge.
By exploiting the linearity of this operator, we will formulate a linear Kalman filter framework, which allows for estimation of trajectories of the underlying system state based on noisy nonlinear measurements.
Now we will briefly overview the key points of our method, which is described in detail in the following subsections.

Given noisy measurements $\boldsymbol{z}_t\in\mathbb{R}^m$, we use diffusion maps to represent the system state $\boldsymbol{\theta}_t$ by a new set of $k$ coordinates, denoted by $\boldsymbol{\Phi}_t$, where $k<m$. We will show that in this new coordinate system, the dynamical system can be described by the following \emph{linear} equations:
\begin{eqnarray}
\boldsymbol{\dot{\Phi}}_t & = & F\boldsymbol{\Phi}_t + Q_t^{1/2}\boldsymbol{\dot{\omega}}_t\label{eq:dynKalmanOverview}\\
\boldsymbol{z}_t & = & H\boldsymbol{\Phi}_t + R_t^{1/2}\boldsymbol{v}_t\label{eq:measKalmanOverview}
\end{eqnarray}
where $F$ is a linear operator describing the linear drift of the dynamics of the new coordinates $\boldsymbol{\Phi}_t$, $H$ is a linear lift operator from the new coordinates $\boldsymbol{\Phi}_t$ to the measurements $\boldsymbol{z}_t$, $\boldsymbol{\dot{\omega}}_t$ is a standard normally distributed noise process, $\boldsymbol{v}_t$ is measurement noise, and $Q_t$ and $R_t$ are matrices determining the covariance of the driving and measurement noise processes, respectively. 

We will further show that the linear operators $F$ and $H$, in \eqref{eq:dynKalmanOverview} and \eqref{eq:measKalmanOverview} respectively, can be constructed using diffusion maps in a data-driven manner, providing a coarse approximation of the system. Yet, this construction can be further improved since diffusion maps ignore the inherent time-dependencies between consecutive samples. 
Therefore, we formulate a linear Kalman filter using the constructed coordinates $\boldsymbol{\Phi}_t$ and the recovered system operators $F$ and $H$, thereby improving the state estimate by incorporating the dynamics into the diffusion maps coordinates. This leads to a data-driven linear filtering framework, which can be applied to nonlinear systems with an unknown model, revealing a new representation of the system $\boldsymbol{\widehat{\Phi}}_t$, which is tightly related to the system state $\boldsymbol{\theta}_t$ as will be described in Subsection \ref{sub:StateRecovery}.
In addition, our framework allows for the estimation of specific realizations of system trajectories based on measurements, in contrast to most existing work on the Stochastic Koopman Operator, which only represent the average time-evolution.

The remainder of this section is described as follows.
In Subsection \ref{sub:StateRecovery} and Subsection \ref{sub:ModelRecovery}, we reiterate the derivations presented in \cite{shnitzer2016} for state and model recovery using diffusion maps \cite{Coifman2006}. In Subsection \ref{sub:Kalman}, we present our proposed Kalman filter framework. In Subsection \ref{sub:obsdet}, we discuss the observability and detectability of the proposed framework and in Subsection \ref{sub:Koopman}, we elaborate on the relation between our framework and the Stochastic Koopman Operator. 

\subsection{Recovering the State\label{sub:StateRecovery}}
Given noiseless measurements, $\boldsymbol{z}_t=g\left(\boldsymbol{\theta}_t\right)$, of the hidden system state, $\boldsymbol{\theta}_t$, the following kernel is defined
\begin{equation}
k_{\epsilon}\left(s,t\right)=\exp\left(-\frac{d^2\left(\boldsymbol{z}_s,\boldsymbol{z}_t\right)}{\epsilon^2}\right)\label{eq:contKern}
\end{equation}
where $\epsilon>0$ is a kernel scale, traditionally set to the median of the distances between the measurements, and $d\left(\boldsymbol{z}_s,\boldsymbol{z}_t\right)$ is a distance function between $\boldsymbol{z}_s$ and $\boldsymbol{z}_t$. In our case, we calculated this distance using a modified Mahalanobis distance, first presented in \cite{Singer2008}:
\begin{equation}
d\left(\boldsymbol{z}_s,\boldsymbol{z}_t\right) = \sqrt{\frac{1}{2}\left(\boldsymbol{z}_s-\boldsymbol{z}_t\right)\left(C^{\dagger}_s+C^{\dagger}_t\right)\left(\boldsymbol{z}_s-\boldsymbol{z}_t\right)^T}\label{eq:Mahalanobis}
\end{equation}
where $C^{\dagger}_s$ and $C^{\dagger}_t$ are the psuedo-inverse of the measurement covariance matrices at times $s$ and $t$, respectively.
In \cite{Singer2008}, it was shown that this modified Mahalanobis distance between the measurements approximates the Euclidean distance between the hidden states.

The constructed kernel is then normalized according to
\begin{equation}
p_{\epsilon}\left(s,t\right)=\frac{k_{\epsilon}\left(s,t\right)}{d_{\epsilon}\left(s\right)}
\end{equation}
where $d_{\epsilon}\left(s\right)=\int k_{\epsilon}\left(s,t\right)p_{eq}\left(\boldsymbol{\theta}_t\right)d\boldsymbol{\theta}_t$ and $p_{eq}\left(\boldsymbol{\theta}_t\right)=e^{-U\left(\boldsymbol{\theta}_t\right)}$ is the equilibrium density of the hidden state $\boldsymbol{\theta}_t$.

We can then define the operator $P_{\epsilon}$ by 
\[
	\left(P_{\epsilon}f\right)\left(\boldsymbol{\theta}_s\right)=\int p_{\epsilon}\left(s,t\right)f\left(\boldsymbol{\theta}_t\right)p_{eq}\left(\boldsymbol{\theta}_t\right)d\boldsymbol{\theta}_t
\]
where $f$ is a real function defined on the hidden state $\boldsymbol{\theta}_t$.
Based on $P_{\epsilon}$ we define
\begin{equation}
L_{\epsilon}=\frac{P_{\epsilon}-I}{\epsilon}\label{eq:OpConv}
\end{equation}
where $I$ is the identity operator. In \cite{Singer2008} it was shown that for hidden states $\boldsymbol{\theta}_t$ with dynamics as in \eqref{eq:Langevin}, the operator $L_{\epsilon}$ converges to the backward Fokker-Planck operator $\mathcal{L}$ defined on the manifold $\mathcal{M}$, as $\epsilon\rightarrow 0$:
\begin{equation}
\mathcal{L}q=\frac{1}{\beta}\Delta q-\nabla q\cdot\nabla U
\end{equation}
where $q$ denotes functions in a subspace of observables defined on the system state, which describe averages of functions: $q\left(\boldsymbol{\theta}_t\right)=\mathbb{E}\left[h\left(\boldsymbol{\theta}_t\right)\mid\boldsymbol{\theta}_0=\boldsymbol{a}_0\right]$, where $\boldsymbol{a}_0$ is some initial state, and $h$ is smooth.

The operator $L_\epsilon$ has a discrete set of decreasing eigenvalues, $\{-\lambda^{(\ell)}\}_{\ell\in\mathbb{N}}$, $0=-\lambda^{(0)}>-\lambda^{(1)}\geq-\lambda^{(2)}\geq ...$, and eigenfunctions, $\{\boldsymbol{\phi}^{(\ell)}\}_{\ell\in\mathbb{N}}$ \cite{Coifman2006}. In many stochastic systems, these eigenvalues exhibit a spectral gap, implying on only a few dominant eigenvalues (which are close to $0$) \cite{coifman2008diffusion}. In such systems, we can represent the hidden state using only the principal eigenfunctions, corresponding to the largest (non-trivial $\ell\neq 0$) eigenvalues.
The diffusion maps coordinates are then obtained by calculating the eigenvalue decomposition of $L_{\epsilon}$ and using the $k$ eigenfunctions corresponding to the $k$ largest eigenvalues:
\begin{equation}
\boldsymbol{z}_t\mapsto\boldsymbol{\Phi}\left(\boldsymbol{\theta}_t\right)=\left[\phi^{(1)}\left(\boldsymbol{\theta}_t\right),\phi^{(2)}\left(\boldsymbol{\theta}_t\right),...,\phi^{(k)}\left(\boldsymbol{\theta}_t\right)\right]\label{eq:DMCoords}
\end{equation}

It was shown in \cite{Singer2008} that when using a kernel based on the modified Mahalanobis distance \eqref{eq:Mahalanobis}, the eigenfunctions of $L_\epsilon$ converge to the eigenfunctions of the backward Fokker-Planck operator $\mathcal{L}$, defined on the \emph{hidden state} $\boldsymbol{\theta}_t$, as $\epsilon \rightarrow 0$.

The diffusion maps coordinates in \eqref{eq:DMCoords} are tightly related to the system state.
Consider the adjoint of $\mathcal{L}$, which is known as the forward Fokker-Planck operator. The forward Fokker-Planck operator exhibits two important properties. First, it describes the evolution of the transition probability density, $p\left(\boldsymbol{\theta},t\mid\boldsymbol{\theta}_0\right)$. Second, it is Hermitian, and therefore, its eigenfunctions
form a basis for the space of real functions of the system state (with the equilibrium density as a measure $p_{eq}\left(\boldsymbol{\theta}\right)=\exp^{-U\left(\boldsymbol{\theta}\right)}$). 
In addition, since the backward and forward operators are adjoint, their eigenfunctions can be normalized to be bi-orthonormal, and the eigenfunctions of the backward operator can be used as a bi-orthonormal basis.
Due to this connection between the eigenfunctions of the backward and forward operators and since the eigenfunctions of the forward operator are tightly related to the system state, the eigenfunctions of the backward Fokker-Planck operator can be used to represent the system state as well. 
Therefore, from this point on, the diffusion maps coordinates, denoted by $\boldsymbol{\phi}^{(\ell)}$, which approximate these eigenfunctions, can be considered as a new set of coordinates representing the hidden system state.

Note that the above results are obtained only when we have access to clean measurements $\boldsymbol{z}_t=g\left(\boldsymbol{\theta}_t\right)$. We will address this issue in Subsection \ref{sub:Kalman}.

\subsection{Recovering the Model\label{sub:ModelRecovery}}
We will now show that the representation of the system state using the diffusion maps coordinates \eqref{eq:DMCoords} can be used to construct a data-driven representation of the system model.

\subsubsection*{The Dynamics of the Diffusion Maps Coordinates}
Consider the state dynamics in \eqref{eq:Langevin} measured through $\boldsymbol{z}_t=g\left(\boldsymbol{\theta}_t\right)$. For such systems, based on It\^{o}'s Lemma, the eigenfunctions of the operator $L_\epsilon$,
obtained using the diffusion maps algorithm, evolve according to a stochastic differential equation (SDE) of known form \cite{coifman2008diffusion}:
\begin{equation}
\dot{\phi}^{(\ell)}\left(\boldsymbol{\theta}\right)=-\lambda^{(\ell)}\phi^{(\ell)}\left(\boldsymbol{\theta}\right)+\sqrt{2}\Vert\nabla_{\boldsymbol{\theta}}\phi^{(\ell)}\left(\boldsymbol{\theta}\right)\Vert_2\dot{\omega}^{(\ell)}\label{eq:DMdyn}
\end{equation}
where $\phi^{(\ell)}\left(\boldsymbol{\theta}\right)$ is the $\ell$th eigenfunction of the operator $L_\epsilon$, $-\lambda^{(\ell)}$ is the corresponding eigenvalue, $\omega^{(\ell)}$ is the $\ell$th coordinate of a multidimensional Brownian motion process, where each coordinate is some linear combination of the Brownian motion coordinates of the underlying process, $u^{(\ell)}$, and $\Vert\cdot\Vert_2$ is the $L^2$ norm.
Note that we omitted the time notation of the state ($\boldsymbol{\theta}_t$) to indicate that the eigenfunctions are dependent only on the state.

This SDE depicts that the obtained diffusion maps coordinates evolve according to a linear drift, $-\lambda^{(\ell)}\phi^{(\ell)}\left(\boldsymbol{\theta}\right)$, and an additional diffusion component. We note that the linear drift component is fully known since we have, using diffusion maps, both the eigenfunctions and the corresponding eigenvalues of the operator $L_\epsilon$, which approximates the backward Fokker-Planck operator.
This linear drift will be used as the new \emph{linear} state dynamics.

\subsubsection*{Constructing the Lift Function}
As stated in Subsection \ref{sub:StateRecovery}, the eigenfunctions of the backward Fokker-Planck operator, approximated by diffusion maps, form a basis for the space of real functions defined on the system state. Therefore, every real function of the system state can be written as an expansion in these eigenfunctions. Specifically, we can represent the measurement function in the following manner:
\begin{equation}
z_t^{(i)} = g^{(i)}\left(\boldsymbol{\theta}_t\right)=\sum_{\ell=0}^{\infty}\alpha_{i,\ell}\phi^{(\ell)}\left(\boldsymbol{\theta}_t\right)\label{eq:liftInf}
\end{equation}
where $\alpha_{i,\ell}=\left\langle z^{(i)},\phi^{(\ell)}\right\rangle_{p_{eq}}=\int_{\mathcal{M}}z^{(i)}_t\phi^{(\ell)}\left(\boldsymbol{\theta}_t\right)p_{eq}\left(\boldsymbol{\theta}_t\right)d\boldsymbol{\theta}_t$.
When the eigenvalues of the backward Fokker-Planck operator exhibit a spectral gap, most of the energy is captured by the $k$ eigenfunctions, corresponding to the $k$ largest eigenvalues. In such cases we can approximate the mapping in \eqref{eq:liftInf} using only these $k$ eigenfunctions:
\begin{equation}
z_t^{(i)}\approx\sum_{\ell=0}^{k}\alpha_{i,\ell}\phi^{(\ell)}\left(\boldsymbol{\theta}_t\right)\label{eq:liftFin}
\end{equation}

We can now write expression \eqref{eq:liftFin} in matrix form:
\begin{equation}
\boldsymbol{z}_t\approx\boldsymbol{\alpha}\boldsymbol{\Phi}\left(\boldsymbol{\theta}_t\right)\label{eq:liftFinFinal}
\end{equation}
where $\boldsymbol{\Phi}\left(\boldsymbol{\theta}_t\right)$ is defined in \eqref{eq:DMCoords} and $\boldsymbol{\alpha}$ is an $m\times k$ matrix in which $\left(\boldsymbol{\alpha}\right)_{i,\ell}=\alpha_{i,\ell}$.

These results imply that through the representation of the measurement function using the diffusion maps eigenfunctions, we obtain a linear mapping between the eigenfunctions, $\boldsymbol{\Phi}\left(\boldsymbol{\theta}_t\right)$, and the system measurements, $\boldsymbol{z}_t$. Thus, we set the linear lift function simply to be $\boldsymbol{\alpha}$.

To conclude this subsection, we note that all of the theoretical derivations above are true for state equations of the form \eqref{eq:Langevin}. However, it is not essential that specifically the state will exhibit such dynamics, but rather that the dynamics of some underlying system parameter are governed by the Langevin equation. In such systems, the diffusion maps coordinates, constructed using the modified Mahalanobis distance \eqref{eq:Mahalanobis}, exhibit useful properties and can be used as a foundation for the proposed time-series filtering framework, as described in the remainder of this paper.
Since many natural phenomena are governed by dynamics that can be modeled using Langevin equation \eqref{eq:Langevin}, our framework is applicable to a wide range of problems.

\subsection{Diffusion Maps Kalman Filter\label{sub:Kalman}}
Based on the Euler-Maruyama method, the dynamics of the diffusion maps coordinates can be discretized to
\begin{eqnarray}
\phi^{(\ell)}\left(\boldsymbol{\theta}_{n+1}\right) & = & \left(1-\lambda^{(\ell)}\Delta t\right)\phi^{(\ell)}\left(\boldsymbol{\theta}_{n}\right)\nonumber\\
& + & \sqrt{2}\Vert\nabla_{\boldsymbol{\theta}}\phi^{(\ell)}\left(\boldsymbol{\theta}_n\right)\Vert_2\Delta\omega_n^{(\ell)}\label{eq:discreteDyn}
\end{eqnarray}
where $\Delta t$ is the time step, $\Delta\omega_t^{(\ell)}$ is a normally distributed noise process and $\nabla_{\boldsymbol{\theta}}$ denotes the gradient according to the system state $\boldsymbol{\theta}_n$.
To approximate the linear drift in \eqref{eq:discreteDyn} from discrete measurements, discrete approximations of $\phi^{(\ell)}\left(\boldsymbol{\theta}_{n}\right)$ and $\lambda^{(\ell)}$ are required.
These approximations can be obtained using the discrete counterpart of diffusion maps as follows.

Given discrete-time measurements $\{\boldsymbol{z}_n\}_{n=1}^N$,
the $N\times N$ kernel matrix is calculated similarly to \eqref{eq:contKern} by:
\begin{equation}
K(i,j)=\exp\left(-\frac{d^2\left(\boldsymbol{z}_i,\boldsymbol{z}_j\right)}{\epsilon^2}\right)\label{eq:discreteKern}
\end{equation}
where $\epsilon>0$ is the kernel scale and $d(\cdot,\cdot)$ is the modified Mahalanobis distance \eqref{eq:Mahalanobis}.

The kernel matrix is then normalized to be row-stochastic
\begin{equation}
P(i,j) = \frac{K(i,j)}{D(i)}\label{eq:discreteNormKern}
\end{equation}
where $D(i)=\sum_{j=1}^N K(i,j)$.

From the eigenvalue-decomposition of $P$ we obtain a set of eigenvectors and eigenvalues, $\{\boldsymbol{\psi}^{(\ell)}\}_{\ell=0}^{k}$, $\{\mu^{(\ell)}\}_{\ell=0}^{k}$.
It was shown in \cite{Coifman2006}, that in this discrete setting, the diffusion maps eigenvectors, $\{\boldsymbol{\psi}^{(\ell)}\}_{\ell=0}^{k}$, approximate the continuous diffusion maps eigenfunctions, $\{\boldsymbol{\phi}^{(\ell)}\}_{\ell=0}^{k}$, discussed in the previous subsections. Moreover, it was shown in \cite{Singer2008}, that the eigenvalues of the discrete diffusion maps algorithm can be used to approximate the eigenvalues of the continuous operator according to $-\lambda^{(\ell)}=\frac{2}{\epsilon}\log\mu^{(\ell)}$.  

We rewrite the state and measurement equations in the discrete setting using the diffusion maps coordinates:
\begin{eqnarray}
\boldsymbol{\Psi}_{n+1} & = & \left(I+\Lambda\right)\boldsymbol{\Psi}_n+Q_n^{1/2}\Delta\boldsymbol{\omega}_n\label{eq:dmState}\\
\boldsymbol{z}_n & = & \boldsymbol{\alpha}\boldsymbol{\Psi}_n+R_n^{1/2}\boldsymbol{v}_n\label{eq:dmMeas}
\end{eqnarray}
where $\boldsymbol{\Psi}_n=[\psi^{(1)}_n,\psi^{(2)}_n,\dots,\psi^{(k)}_n]$, $\Lambda$ is a diagonal matrix with $\{-\lambda^{(\ell)}\Delta t\}_{\ell=1}^{k}$ as its diagonal elements, $I$ is the identity matrix, $Q_n^{1/2}$ is a matrix containing the coefficients of the second term in \eqref{eq:discreteDyn}, $\Delta\boldsymbol{\omega}_n$ is a vector of standard normally distributed noise processes, $R_n$ is the covariance of the measurement noise and $\boldsymbol{\alpha}$ is the lift function from the diffusion maps coordinates to the measurements. 
For simplicity of notation, we denote $\boldsymbol{\Psi}_n=\boldsymbol{\Psi}(\boldsymbol{\theta}_n)$, omitting the dependence on $\boldsymbol{\theta}_n$.
Importantly, note that these system equations are approximately linear, even for highly nonlinear systems. Therefore, using the diffusion maps coordinates, we obtain a virtual system state which linearizes the problem, in an entirely data-driven manner, given only the measurements.
In the discrete setting, the lift function is approximated by $\alpha_{i,\ell}=\sum_{n=1}^{N}z^{(i)}_n\psi^{(\ell)}\left(\boldsymbol{\theta}_n\right)$ \cite{shnitzer2016}. 

The term $\boldsymbol{\alpha}\boldsymbol{\Psi}_n$ in \eqref{eq:dmMeas} approximates the nonlinear measurement function $g(\boldsymbol{\theta}_n)$. Therefore, the noise term in \eqref{eq:dmMeas}, denoted by $R_n^{1/2}\boldsymbol{v}_n$, is approximately the measurement noise from \eqref{eq:MeasEq}.
However, since the diffusion maps coordinates are constructed based on the noisy measurements, $\boldsymbol{z}_n$, the approximation of the measurement function may be inaccurate, leading to non-Gaussian components in $R_n^{1/2}\boldsymbol{v}_n$.
In regard to this issue, we note that in many cases the leading diffusion maps eigenvectors (corresponding to the largest eigenvalues) represent the slow components in the data, and therefore, are only mildly affected by the noise. This property facilitates the use of our devised Kalman model in many applications, as demonstrated in the experimental study in Section \ref{sec:Results}, even though the model is inaccurate due to the measurement noise.
We further note that the noise term in \eqref{eq:dmMeas} can be used to represent deviations from the true model as described in \cite{bulut2012kalman}.
However, significant measurement noise may still lead to some deterioration in the performance of the proposed method, as demonstrated in Subsection \ref{sub:Toy1}.

Due to the linearity of the derived system in \eqref{eq:dmState} and \eqref{eq:dmMeas} (except for elements in $Q_n$ as discussed in a subsequent paragraph), we can construct a linear Kalman filter based on the diffusion maps coordinates. 
Using the Kalman filter framework we harness the linear dynamics of the diffusion maps eigenfunctions to improve the approximation of the eigenvectors of the discrete diffusion maps algorithm.
Therefore, our framework incorporates the inherent time-dependencies of the system samples into a manifold learning technique.

Our proposed framework is highly related to the method presented in \cite{shnitzer2016}, where a linear observer was constructed based on the system equations described in \eqref{eq:dmState} and \eqref{eq:dmMeas}, by exploiting the linearity of the data-driven representation in a similar manner.
The Kalman filtering framework improves the linear observer framework in \cite{shnitzer2016} in two main aspects.
First, the constructed observer scheme is a deterministic framework which discards the stochastic term of the dynamics, whereas the Kalman filter takes it naturally into account. Second, the Kalman filter provides an adaptive optimal update of the fixed model parameters in \cite{shnitzer2016}. 
We empirically show in Section \ref{sec:Results}, that the adaptive parameter update and the stochastic framework significantly improve the state recovery and robustness to noise, and outperforms the observer framework as well as other competing methods.

We note that $Q_n$ is a nonlinear function of the state and induces dependencies between different time-samples of the system noise process, which does not fit the Kalman filter framework. However, in many applications, the leading diffusion maps eigenvectors, which are used as a low-dimensional representation of the system, are slowly varying functions of the system state. Therefore, the gradient of $\boldsymbol{\psi}^{(\ell)}$ according to $\boldsymbol{\theta}_n$ in \eqref{eq:discreteDyn} may be sufficiently small (or approximately constant), allowing for the proper use of the Kalman filter. These properties are demonstrated in simulations in Section \ref{sec:Results}.

An alternative approach to address the dependence of the driving noise on the state would be using a particle filter instead of the Kalman filter. Particle filters are designed to support a wider range of driving noise distributions \cite{Arulampalam2002,doucet2000sequential}. However, our empirical study showed that the Kalman filter significantly outperformed the particle filter, defined in a corresponding manner using the derived model, since the particle filter is sensitive to errors in the estimation of the gradient of $\boldsymbol{\psi}^{(\ell)}$.
Improving our methods by appropriately approximating the driving noise remains a subject for future work.

The implemented Kalman filter framework is described as follows:
\begin{eqnarray}
\nonumber
\widehat{\boldsymbol{\Psi}}_{n} & = & F_{n}\widehat{\boldsymbol{\Psi}}_{n-1}+\boldsymbol{\kappa}_{n}\left(\boldsymbol{z}_{n}-H_{n}F_{n}\widehat{\boldsymbol{\Psi}}_{n-1}\right)\\
P_{n} & = & \left(I-\boldsymbol{\kappa}_{n}H_{n}\right)\left(F_{n}P_{n-1}F^T_{n}+Q_{n}\right)\label{eq:Kalman}\\
\boldsymbol{\kappa}_{n} & = & \left(F_{n}P_{n-1}F^T_{n}+Q_{n}\right)H^T_{n}\nonumber\\
& & \left(H_{n}F_{n}P_{n-1}F^T_{n}H^T_{n}+H_{n}Q_{n}H^T_{n}+R_{n}\right)^{-1}\nonumber
\end{eqnarray}
where $\widehat{\boldsymbol{\Psi}}_{n}$ is the state estimate at time $n$, $\boldsymbol{z}_n$ is the measurement at time $n$, $I$ is the identity matrix, $F_{n}=I+\Lambda$ represents the dynamics of the state, $H_{n}=\boldsymbol{\alpha}$ is the lift function between the measurements and the state calculated in a data driven manner as described after equation \eqref{eq:dmMeas}, $Q_{n}$ is the covariance matrix of the state driving noise, and $R_{n}$ is the covariance matrix of the measurement noise.
We refer to the proposed Kalman filter framework as the Diffusion Maps Kalman (DMK).

In the experiments presented in Section \ref{sec:Results}, the covariance matrices, $R_n$ and $Q_n$, were estimated from the data according to $R_n(k,k)=\mathrm{var}\left(\boldsymbol{z}^{(k)}\right)$ and $Q_n(k,k)=\mathrm{var}\left(\lambda^{(k)}\psi^{(k)}\right)$, respectively, where $\psi^{(k)}$ is the $k$th diffusion maps eigenvector corresponding to eigenvalue $\mu^{(k)}$ and $-\lambda^{(k)}=2/\epsilon\log\mu^{(k)}$.
This empirical choice led to good results in all applications. However, it is possibly an underestimation of the covariance and could be further improved. For example the adaptive estimation for $R_n$ and $Q_n$ described in \cite{berry2013adaptiveEK} could be used for unknown system dynamics and measurement function. 
In addition, this issue will be addressed in future work, where a method for properly estimating the variance of the eigenfunctions in \eqref{eq:DMdyn} will be devised.

\subsection{Observability and Detectability of the New System Equations\label{sub:obsdet}}

In this subsection, we discuss the conditions under which the proposed discrete linear system in \eqref{eq:dmState} and \eqref{eq:dmMeas} is observable and detectable.

To address the observability and detectability conditions, we first recall that in the proposed system, the state transition matrix is given by $F=I+\Lambda$, which is a full rank diagonal matrix, and the observation matrix is given by $H=\boldsymbol{\alpha}$, where $\boldsymbol{\alpha}$ is defined in Subsection \ref{sub:ModelRecovery}, after equation \eqref{eq:liftInf}.

\begin{prop}
The system in \eqref{eq:dmState} and \eqref{eq:dmMeas} is observable if $\forall\ell=1,\dots,k, \ \exists p=1,\dots,m$ such that $\left\langle \boldsymbol{z}^{(p)},\boldsymbol{\psi}^{(\ell)}\right\rangle\neq 0$, where $\left\langle \boldsymbol{z}^{(p)},\boldsymbol{\psi}^{(\ell)}\right\rangle=\sum_{n=1}^N z_n^{(p)}\psi^{(\ell)}(\boldsymbol{\theta}_n)$.
\end{prop}
In other words, for observability it is required that the $k$ diffusion maps eigenvectors, used in the construction of the system, are not orthogonal to all measurement coordinates.
\begin{proof}
In the discrete and linear Kalman filter setting, a system is observable if there are no vectors $\boldsymbol{v}^{(\ell)}\neq\boldsymbol{0}$ such that $F\boldsymbol{v}^{(\ell)}=\gamma^{(\ell)}\boldsymbol{v}^{(\ell)}$ and $H\boldsymbol{v}^{(\ell)}=\boldsymbol{0}$, where $\boldsymbol{v}^{(\ell)}$ and $\gamma^{(\ell)}$ denote the $\ell$th eigenvector and corresponding eigenvalue of $F$, respectively \cite{sontag2013mathematical}.

Since $F$ is a diagonal matrix in \eqref{eq:dmState}, its eigenvectors, $\boldsymbol{v}^{(\ell)}$, contain $1$ at index $\ell$ and zeros elsewhere.
Therefore, $H\boldsymbol{v}^{(\ell)}$ is the $\ell$th column of $H$, which based on Subsection \ref{sub:ModelRecovery}, corresponds to the following vector:
\begin{equation}
H\boldsymbol{v}^{(\ell)}=
\begin{bmatrix}
\left\langle \boldsymbol{z}^{(1)},\boldsymbol{\psi}^{(\ell)}\right\rangle\\
\vdots\\
\left\langle \boldsymbol{z}^{(m)},\boldsymbol{\psi}^{(\ell)}\right\rangle
\end{bmatrix}
\end{equation}
where $\boldsymbol{z}^{(i)}$ denotes the $i$th measurement coordinate, $\boldsymbol{\psi}^{(\ell)}$,$\ell=1,\dots,k$, is the $\ell$th diffusion maps eigenvector and $\left\langle\cdot,\cdot\right\rangle$ denotes the inner product between two vectors. 
From this derivation, for observability we require that $Hv^{(\ell)}\neq 0$, $\forall	\ell=1,\dots,k$, i.e. that the diffusion maps coordinates are not orthogonal to all measurement coordinates.
\end{proof}

\begin{prop}
If $\exists\ell=1,\dots,k$, such that $\left\langle \boldsymbol{z}^{(p)},\boldsymbol{\psi}^{(\ell)}\right\rangle=0$, $\forall p=1,\dots,m$, then the system is detectable if $1-\lambda^{(\ell)}\Delta t<0$, where $-\lambda^{(\ell)}$ is the $\ell$th eigenvalue of the continuous operator, approximated using diffusion maps, $\boldsymbol{\psi}^{(\ell)}$ is the corresponding eigenvector and $\Delta t$ denotes the time step.
\end{prop}
\begin{proof}
A discrete linear system is detectable, if there is no $v^{(\ell)}\neq 0$ and $\gamma^{(\ell)}$ such that $Fv^{(\ell)}=\gamma^{(\ell)}v^{(\ell)}$, $Hv^{(\ell)}=0$ with $\gamma^{(\ell)}+(\gamma^{(\ell)})^*\geq 0$, where $()^*$ denotes the complex conjugate \cite{sontag2013mathematical}.

In the constructed system, the eigenvalues, $\gamma^{(\ell)}$, of $F$ are real and equal to the diagonal elements of $F$, which take the form of $1-\lambda^{(\ell)}\Delta t$.
Therefore, the constructed system is detectable if $1-\lambda^{(\ell)}\Delta t<0$, when $\left\langle \boldsymbol{z}^{(p)},\boldsymbol{\psi}^{(\ell)}\right\rangle=0$, $\forall p=1,\dots,m$, i.e. $\gamma^{(\ell)}<0$ for $\psi^{(\ell)}$, for which $Hv^{(\ell)}=0$.
\end{proof}

Note that by using only diffusion maps coordinates that are not orthogonal to all system measurements in the construction of the system, both propositions hold. 
Therefore, an informed choice of the diffusion maps coordinates leads to an observable and detectable system.
According to common practice, we take the leading $k$ diffusion maps eigenvectors, corresponding to the largest eigenvalues, since these coordinates typically represent the system and are only mildly affected by the noise.
Moreover, the choice of the new system state dimensionality, $k$, can be motivated by the two propositions, since we expect that eigenvectors which \emph{are} orthogonal to the measurements, e.g. due to noise, will correspond to smaller eigenvalues, i.e. will have larger indexes $\ell$.
Empirically, following this procedure led to observable and detectable systems in the experimental study in Section \ref{sec:Results}.

\subsection{Koopman Operator\label{sub:Koopman}}

Our presented framework is tightly related to the Koopman Operator. Specifically, we show the analogy of the revealed dynamics in the present work to the Stochastic Koopman Operator \cite{mezic2005spectral}, $\left(U_{st}f\right)\left(\boldsymbol{\theta}_{n}\right)=\mathbb{E}\left[f\circ T\left(\boldsymbol{\theta}_{n},\boldsymbol{\omega}_n\right)\right]$.
Given measurements from some stochastic nonlinear system with state dynamics of the form of \eqref{eq:Langevin}, we project the problem onto the eigenvector space obtained by diffusion maps. By taking these eigenvectors as observables we obtain a space in which the evolution of the observables is represented by a known linear operator:
\begin{align}
\left(\tilde{U}_{st}\psi^{(\ell)}\right)\left(\boldsymbol{\theta}_{n}\right)=\mathbb{E}&\left[\psi^{(\ell)}\circ T\left(\boldsymbol{\theta}_{n},\boldsymbol{u}_n\right)\mid\boldsymbol{\Psi}_n\right]\nonumber\\
=\mathbb{E}&\Biggl[\left(1-\lambda^{(\ell)}\Delta t\right)\psi^{(\ell)}\left(\boldsymbol{\theta}_{n}\right)\nonumber\\
& +\sqrt{2}\Vert\nabla_{\boldsymbol{\theta}}\psi^{(\ell)}\left(\boldsymbol{\theta}_n\right)\Vert_2\Delta\omega_n^{(\ell)}\mid\boldsymbol{\Psi}_n\Biggr]\nonumber\\
= & \left(1-\lambda^{(\ell)}\Delta t\right)\psi^{(\ell)}\left(\boldsymbol{\theta}_n\right)
\end{align}
where $\tilde{U}_{st}$ is analogous to the Stochastic Koopman Operator. Note that in contrast to the standard Stochastic Koopman Operator, $\tilde{U}_{st}$ is conditioned on $\boldsymbol{\Psi}_n$.

The use of diffusion maps for the approximation of the Koopman Operator is discussed in detail in \cite{giannakis2015data}. In particular a Galerkin method for approximating the eigenfunctions of the Koopman generator using the diffusion maps eigenfunctions and eigenvalues is presented. Furthermore, it was shown that for ergodic systems with pure-point spectra, the eigenfunctions of the Koopman generator can be robustly estimated from finite data using diffusion maps.
Importantly, our method is completely different than the method presented in \cite{giannakis2015data}, since we approximate the Stochastic Koopman Operator, whereas in \cite{giannakis2015data} the generator is approximated.

In our proposed framework, we combine the constructed linear operator and observables, obtained by diffusion maps, with a Kalman filter. This leads to two main benefits. First, instead of representing the average time evolution of the observables in stochastic systems, we obtain an estimation of individual realizations of trajectories from the measurements. In addition, the Kalman filter compensates for the noise and deviations from the measurements. Second, due to the use of diffusion maps, we obtain a data-driven dimensionality reduction and approximate the Koopman Operator based on a finite set of orthonormal functions, spanning the state space of the system \cite{coifman2008diffusion}.

A related work, combining the Koopman Operator and a Kalman filter is presented in \cite{surana2016linear}. There, the authors define the Koopman Observer Form (KOF) for noiseless systems and the Koopman Kalman Filter (KKF) for systems with measurement noise. They construct a set of linear update equations based on the eigenvectors and modes of the Koopman Operator which provides a linear filtering framework for nonlinear systems, where the Koopman Operator of a given data-set is approximated using EDMD \cite{Williams2015edmd}. The EDMD algorithm requires a dictionary of basis functions which affects the resulting estimations \cite{Williams2015edmd}. Conversely, in the proposed work, we obtain the linear dynamics and observables based on the data, from the diffusion maps algorithm without a pre-defined dictionary. Another difference is the problem setting, which includes in the present work stochastic system dynamics rather than measurement noise only. We compare our proposed framework to the one presented in \cite{surana2016linear} in Subsection \ref{sub:Toy1}.

Our framework can also be partially related to the work presented in \cite{brunton2017HAVOK}. There, linear update equations are learned from the data for chaotic systems using concepts from Koopman theory. In order to represent the chaotic dynamics using a finite approximation, a nonlinear forcing term representing the deviation from linearity is added. In our work, we rely on the dynamics of the diffusion maps eigenfunctions, which can be expressed as a sum a linear drift component and a nonlinear stochastic component. The stochastic component represents deviations from the simple linear dynamics and can be used for error analysis as well. In the proposed Kalman filter framework we consider the stochastic component as the system noise.

\section{Experimental Results\label{sec:Results}}

In this section we present two simulated examples of object tracking and a real tracking application based on neuronal spiking activity. All three tracking problems are nonlinear with unknown system dynamics and measurement functions, where each example depicts a different measurement modality.
We compare our proposed Diffusion Maps Kalman (DMK) to several competing algorithms, which are detailed in the following. We show that our DMK framework leads to improved state estimates compared with \emph{non-parametric algorithms}. In addition, it obtains results which are comparable with \emph{parametric methods}, which, in contrast to DMK, are provided with the system model.

\subsection{Nonlinear Object Tracking\label{sub:Toy1}}

We first present a model with Gaussian noise, where the location of an object in a 2-dimensional space is measured through its radius and azimuth angle. The underlying process, describing the Cartesian position of the object at each time point is given by the following discrete time Langevin equations:
\begin{eqnarray}
\Delta\theta^{(1)}_{n+1} & = -\frac{1}{2}\left(\theta_n^{(1)}-1\right)^3+\left(\theta_n^{(1)}-1\right) + \sqrt{2}u^{(1)}_n\nonumber\\
\Delta\theta^{(2)}_{n+1} & = -\frac{1}{2}\left(\theta_n^{(2)}-6\right)^3+\left(\theta_n^{(2)}-6\right) + \sqrt{2}u^{(2)}_n\label{eq:Toy1sys}
\end{eqnarray}
where $u^{(1)}_n$ and $u^{(2)}_n$ denote standard Gaussian noises and the drift terms in these equations describe double-well potentials.
An example of the resulting 2-dimensional process is presented in Figure \ref{fig:Toy1process}

\begin{figure}
\centering
\includegraphics[width=0.4\textwidth]{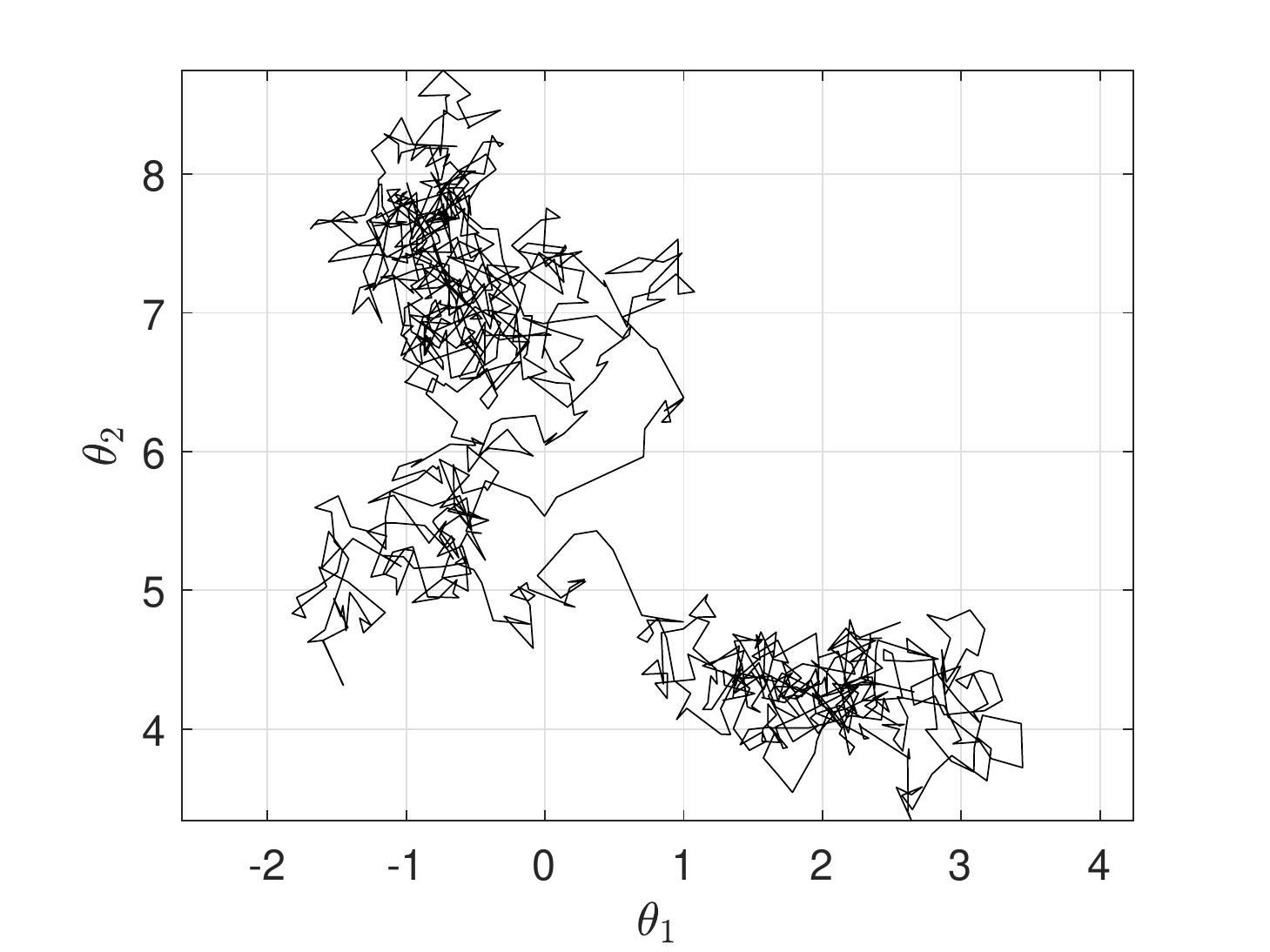}
\caption{Example of the underlying 2-dimensional process.}
\label{fig:Toy1process} 
\end{figure}

The object location is measured in polar coordinates, through the azimuth and radius:
\begin{eqnarray}
\phi_n & = & \mathrm{arc}\tan\left(\frac{\theta^{(1)}_n}{\theta^{(2)}_n}\right)\nonumber\\
r_n & = & \sqrt{\left(\theta^{(1)}_n\right)^2+\left(\theta^{(2)}_n\right)^2}\label{eq:Toy1meas}
\end{eqnarray}
and the system measurements are created by adding Gaussian noise, $\boldsymbol{z}_n=\left[\phi_n+v^{(\phi)}_n,r_n+v^{(r)}_n\right]$, where $v^{(\phi)}_n$ and $v^{(r)}_n$ are Gaussian noise processes with variance $\sigma^2_{\phi}$ and $\sigma^2_{r}$, respectively.
We created trajectories of $1000$ samples with a time-step of $\Delta t=0.01$ and different Signal-to-Noise (SNR) ratios, by varying $\sigma^2_{\phi}$ and $\sigma^2_{r}$.

We applied the diffusion maps algorithm to the measurements, $\boldsymbol{z}_n$, with $\epsilon$ set as the $median$ of the distances.
The covariance matrices used in the modified Mahalanobis distance \eqref{eq:Mahalanobis} were estimated similarly to previous work \cite{TalmonPNAS}, by calculating the empirical covariance of overlapping windows of $30$ time frames, centered at each measurement, i.e. $C_n = \mathrm{cov}(\boldsymbol{z}_{n-N:n+N-1})$ where $N = 15$. Since the covariance matrices are not necessarily full rank, their pseudoinverse was calculated using singular value decomposition (SVD).
We then constructed the Diffusion Maps Kalman (DMK) based on the first two largest eigenvalues and corresponding eigenvectors obtained by diffusion maps. The dimensionality of the diffusion maps coordinates was determined based on the existence of a significant spectral gap after the second coordinate.
The covariance matrices in the Kalman filter update equations, $R_n$ and $Q_n$ in \eqref{eq:Kalman}, were estimated from the data, according to $R_n(k,k)=\mathrm{var}\left(\boldsymbol{z}^{(k)}\right)$ and $Q_n(k,k)=\mathrm{var}\left(\lambda^{(k)}\psi^{(k)}\right)$, where $-\lambda^{(k)}=2/\epsilon\log\mu^{(k)}$ and $\mu^{(k)}$ and $\psi_k$ are the $k$th diffusion maps eigenvalue and eigenvector, respectively.

Using this setting, we evaluated our DMK algorithm, which requires access only to the noisy measurements, $\boldsymbol{z}_n$, and compared it to a parametric algorithm, the particle filter (PF).
For evaluation purposes, the particle filter was provided with the \emph{true system model}, which is considered unknown in our setting.

In addition, we compared our results to three non-parametric algorithms: Gaussian Process filtering (GP) \cite{deisenroth2012robust}, the Kalman filter based algorithm described in \cite{surana2016linear} (KKF) and the observer framework presented in \cite{shnitzer2016}. Note that there is a fundamental difference between our method and two of the competing non-parametric methods, GP and KKF. Both GP and KKF require a subset of data pairs, $\left\{\boldsymbol{\theta}_n,\boldsymbol{z}_n\right\}_{n=1}^{N}$, in their construction.
In our experiments, we used a subset of $N=100$ and $N=21$ samples, covering the entire state space, for algorithms GP and KKF, respectively. 
Moreover, the KKF algorithm requires a choice of a kernel function. Here, we used the same kernel function as in \cite{surana2016linear}.
In contrast, our DMK framework can provide a filtered version of the measurements and a new set of coordinates representing the system characteristics without any information on the underlying state or a specific choice of a kernel function.
However, note that if an estimate of a \emph{specific state representation} is required, some alignment between the DMK coordinates and the underlying state may be needed.

Here, we used the mapping defined in \eqref{eq:liftFinFinal}, $\boldsymbol{\alpha}$, and obtained an estimate of the clean system measurement, $g\left(\boldsymbol{\theta}_n\right)$. 
Therefore, the comparison between the different algorithms was performed in the measurement domain.

\begin{figure}
\centering
\includegraphics[width=0.5\textwidth]{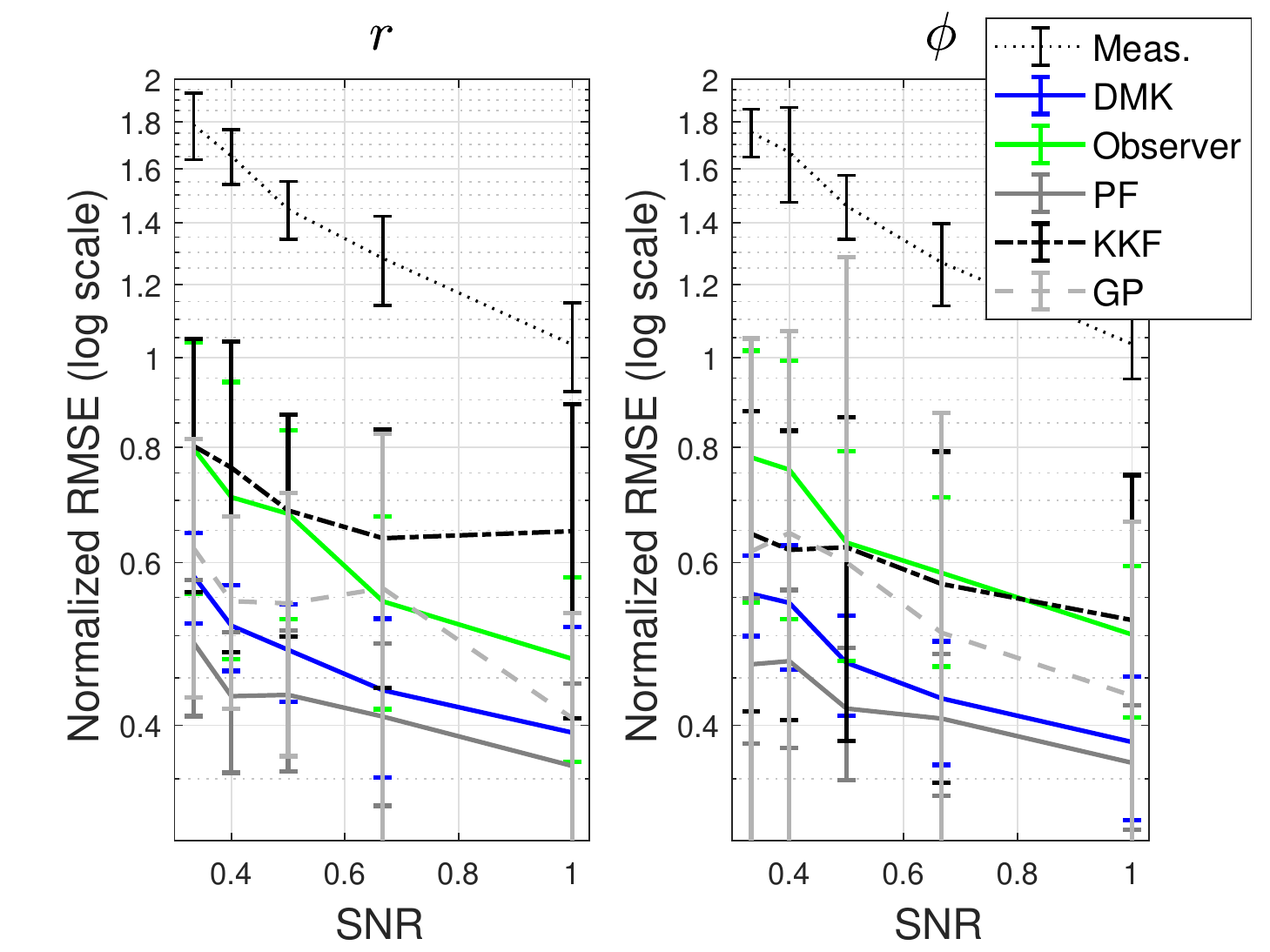}
\caption{Average and standard deviation of the nRMSE of the clean measurement estimates. The nRMSE values were averaged over $50$ realizations of process and noise trajectories.}
\label{fig:Toy1RMSE}
\end{figure}

Figure \ref{fig:Toy1RMSE} presents plots of the normalized root mean square error (nRMSE) values (in log scale) of the clean measurement estimates, $\phi_n$ and $r_n$, obtained by the compared algorithms and the measurement error (denoted by Meas.) as an upper bound. 
The nRMSE values are presented as a function of the measurement noise level (SNR), where the average and standard deviation of the nRMSE were calculated over $50$ different process realizations, for each noise level separately.
We note that the extended Kalman filter (EKF), with the true system model, was considered as well. However, the EKF led to results which were similar to the particle filter and was omitted for brevity therefor.
Figure \ref{fig:Toy1RMSE} depicts that the DMK outperforms all three non-parametric algorithms in all noise levels.
In addition, for high SNR, the DMK errors are close to the errors obtained by the parametric particle filter, \emph{which has access to the true system equations}.
Figure \ref{fig:Toy1trace} further demonstrates this and presents an example for the clean measurement estimation obtained by DMK (in blue) and by the particle filter (in gray), compared with the true clean measurements (in dotted black) and the noisy measurements (gray 'x'), with SNR=1. Plot (a) presents coordinate $\phi_n$ and plot (b) presents coordinate $r_n$. This figure depicts that the DMK estimation closely follows the particle filter estimation.

\begin{figure}
\centering
\subfloat[]{\includegraphics[width=0.45\textwidth]{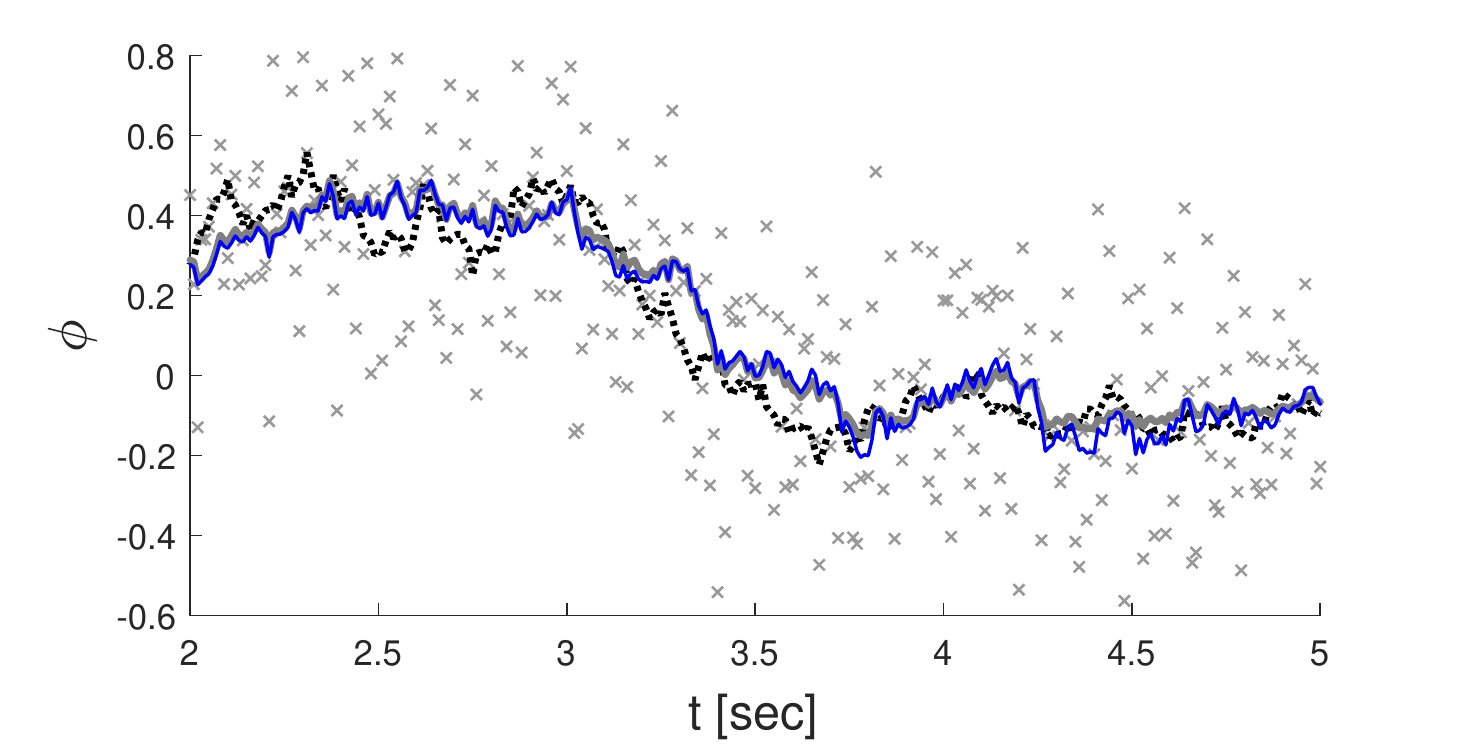}
}

\subfloat[]{\includegraphics[width=0.45\textwidth]{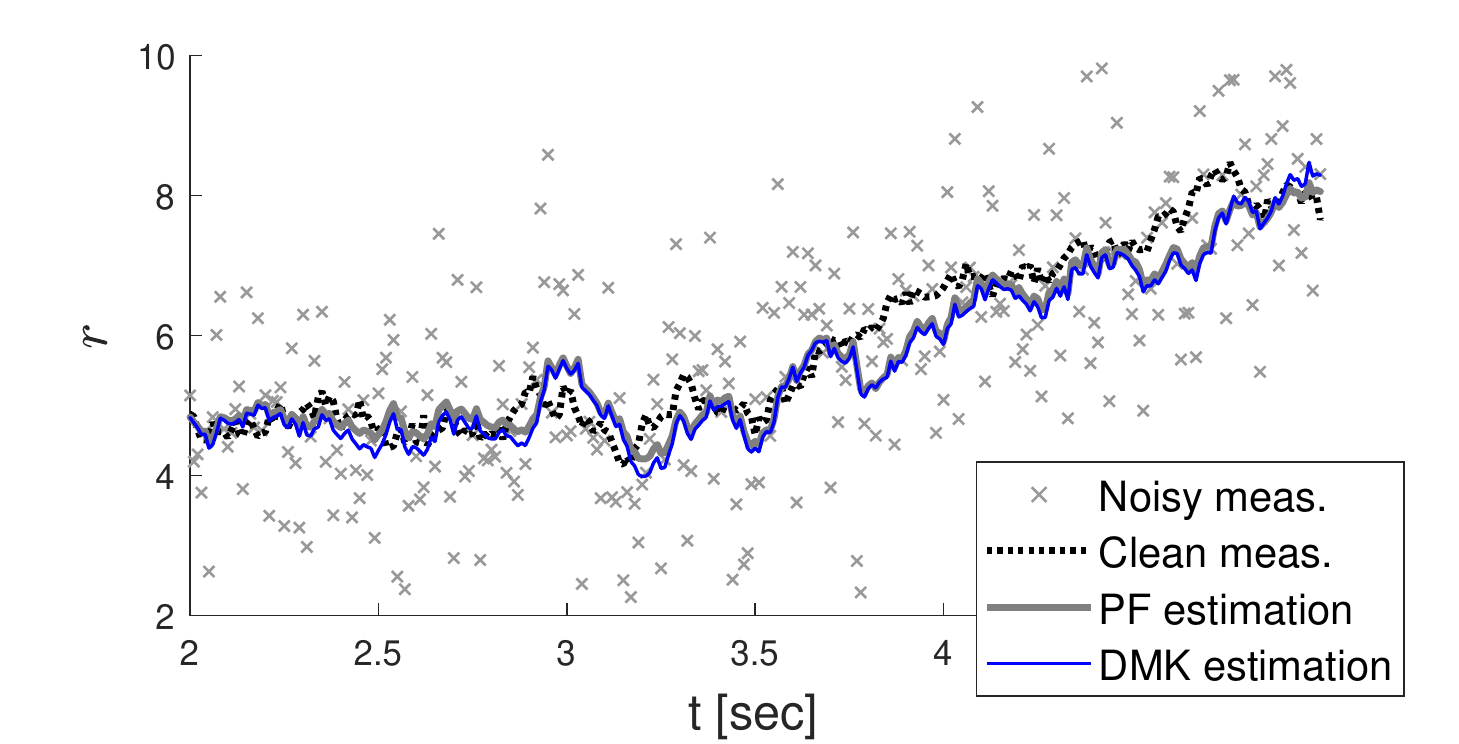}
}
\caption{Example of a trajectory of the filtered measurement coordinates using the DMK algorithm (in blue) compared with the PF estimation (in gray), the true clean measurements (in dotted black) and the noisy measurements (gray 'x').}
\label{fig:Toy1trace}
\end{figure}

Note that for lower SNR values the DMK estimation deteriorates in comparison with the particle filter. This is due to errors in the inferred model, caused by the high noise.
As described in Subsection \ref{sub:StateRecovery} and Subsection \ref{sub:ModelRecovery}, the new system model, which is derived based on diffusion maps, is accurate only for the noiseless case. 
The noise affects the modified Mahalanobis distance and, as a result, the obtained diffusion maps coordinates and eigenvalues contain errors. 
Since the eigenvalues are used in the model dynamics, these errors affect the estimation quality of our framework.
Nevertheless, DMK works quite well in noisy situations, especially in comparison with the non-parametric frameworks.

In order to demonstrate the extent of the estimation deterioration due to noise, we compared the DMK algorithm, which was constructed based on the noisy measurements, and a modified DMK algorithm, where the dynamics (diffusion maps eigenvalues) were obtained by applying diffusion maps to the \emph{clean} measurements and calculating the corresponding ``clean dynamics''.
The resulting nRMSE for different SNR values are presented in Figure \ref{fig:Toy1noise_analysis}, where the average and standard deviation of the nRMSE were calculated based on $50$ process and noise realizations for each SNR value. This figure presents the DMK based on the noisy measurements in blue, the DMK with the ``clean dynamics'' ($\lambda$) in cyan and the particle filter, as a baseline, in gray.
As expected, for high SNR, both the noisy dynamics and the clean dynamics lead to a similar result, whereas for low SNR, the clean dynamics improve the result.
This indicates that our method is indeed affected by the measurement noise. 
However, note that even with these model errors, our method still outperforms the competing non-parametric algorithms, as depicted in Figure \ref{fig:Toy1RMSE}.

\begin{figure}
\includegraphics[width=0.5\textwidth]{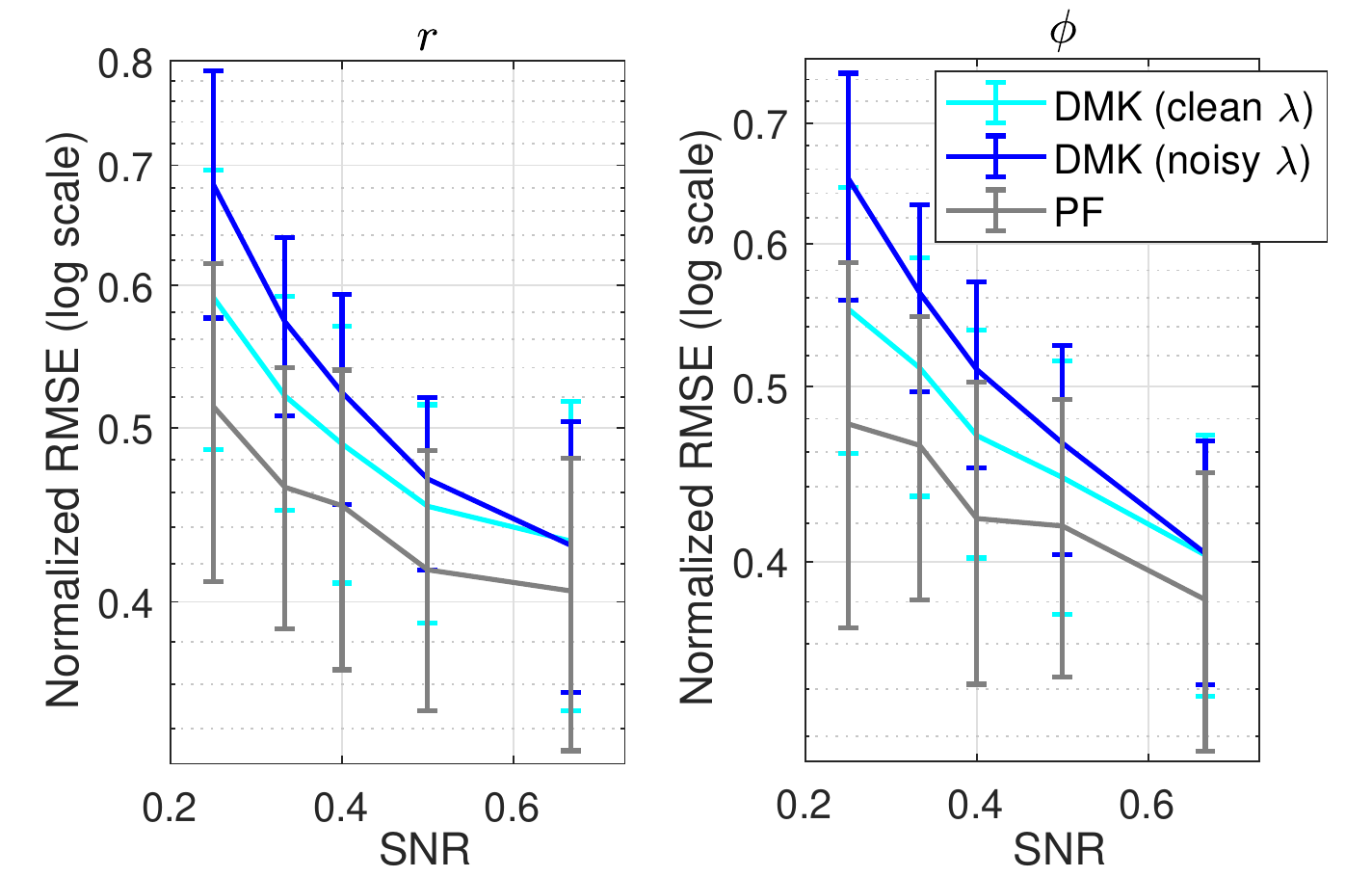}
\caption{Average and standard deviation of the nRMSE of the clean measurement estimates, over $50$ process and noise realizations. A comparison between the DMK algorithm constructed based on the noisy measurements (in blue) the DMK with the ``clean dynamics'' (in cyan) and the PF (in gray).}
\label{fig:Toy1noise_analysis}
\end{figure}

To complete the analysis of this example, we compare the convergence rates of DMK based on the noisy measurements, DMK with the ``clean dynamics'' and the convergence rates of the parametric and non-parametric algorithms.

Figure \ref{fig:Toy1convergence1} and Figure \ref{fig:Toy1convergence2} present the average asymptotic RMSE (aRMSE), i.e. $\sqrt{\frac{1}{M}\sum_{k=1}^M\left\Vert\boldsymbol{\hat{z}}_n^{(k)}-\boldsymbol{\zeta}_n^{(k)}\right\Vert^2_2}$, where $M$ is the number of process and noise realizations, $M=50$ in both figures, $\boldsymbol{\hat{z}}_n^{(k)}$ denotes an estimation of the measurements (filtered) and $\boldsymbol{\zeta}_n^{(k)}$ denotes the true clean measurements at time $n$ in realization $k$.
Figure \ref{fig:Toy1convergence1} presents a comparison between the DMK algorithm based on the noisy measurements (in blue), the DMK algorithm with the ``clean dynamics'' (in cyan) and the particle filter (in gray), for two SNR values. Plot (a) presents the aRMSE values for measurement coordinate $\phi_n$ and plot (b) presents the aRMSE values for measurement coordinate $r_n$. 
This figure depicts that the convergence rate of DMK is not affected by the noise. In addition, the convergence rate of DMK is either comparable (in plot (a)) or faster (in plot (b)) compared with the convergence rate of the particle filter, for both SNR values.

\begin{figure}
\centering
\subfloat[]{\includegraphics[width=0.5\textwidth]{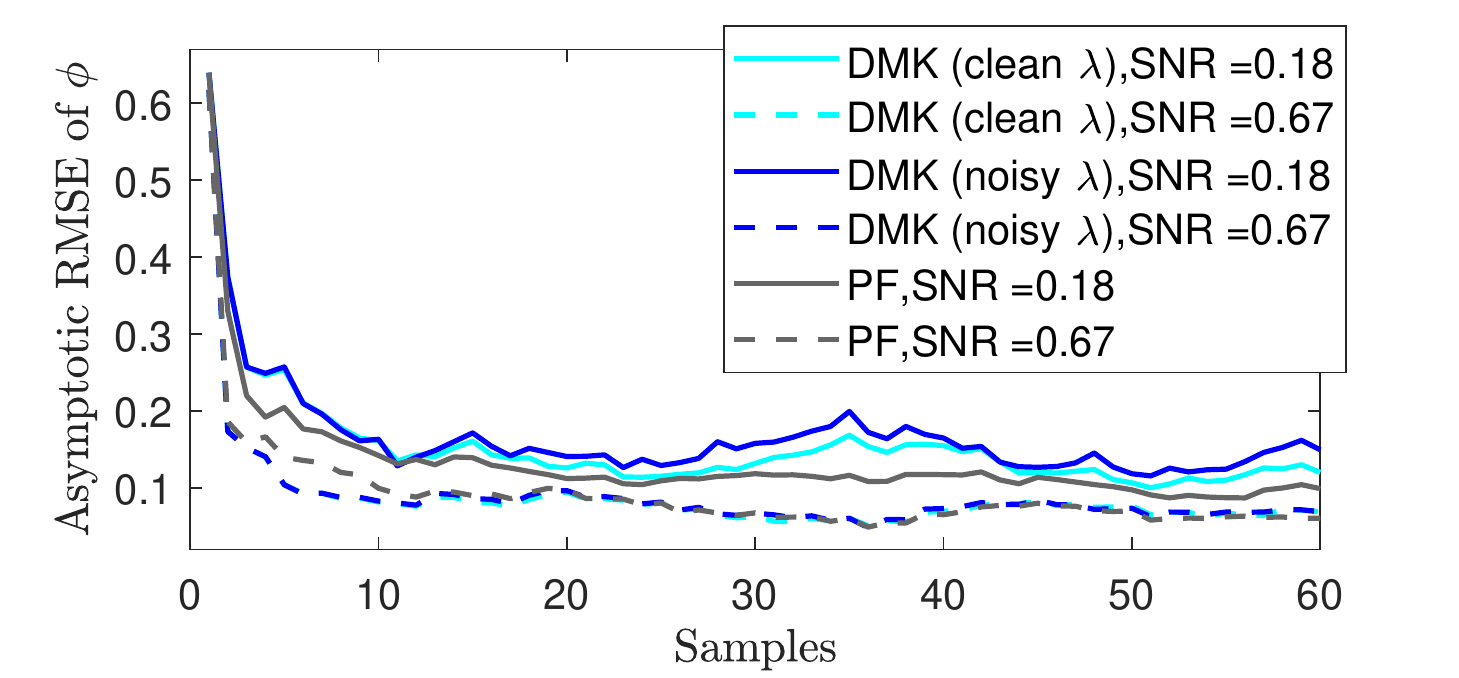}
}

\subfloat[]{\includegraphics[width=0.5\textwidth]{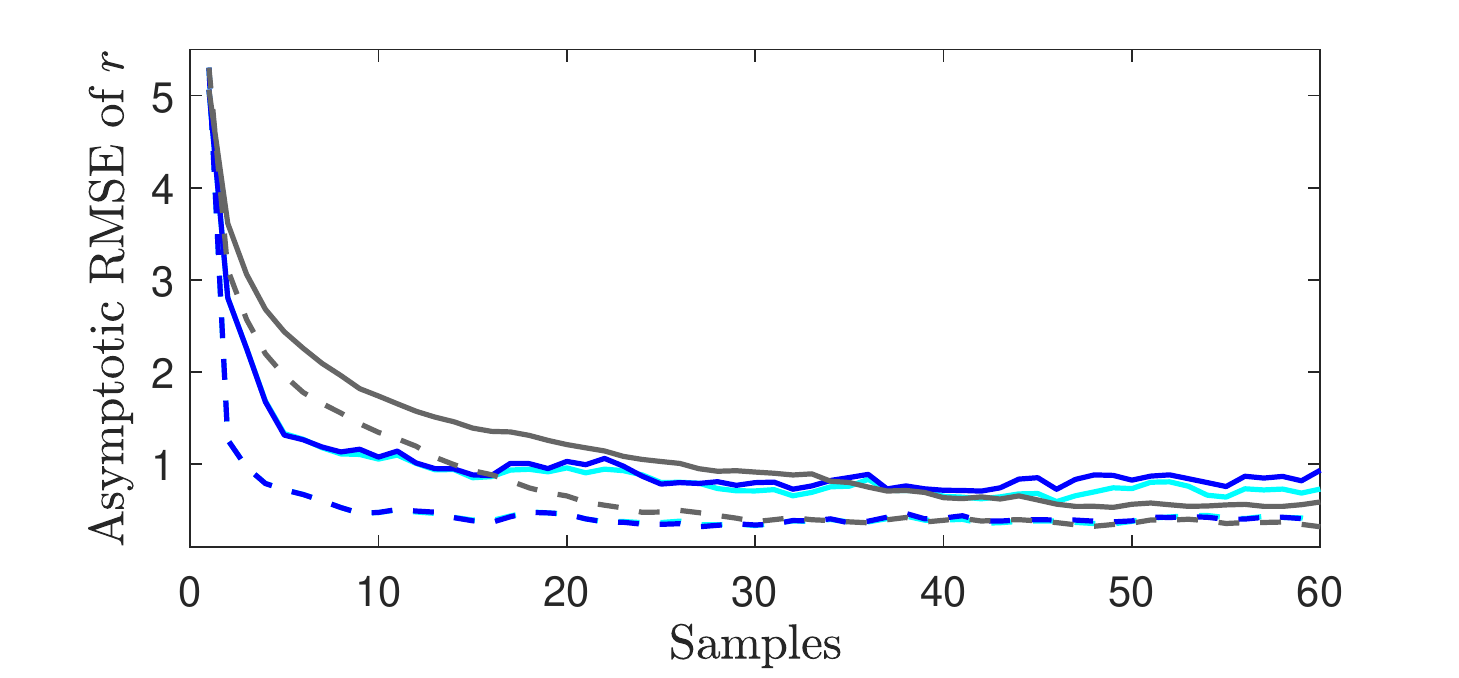}
}
\caption{Asymptotic RMSE, averaged over $50$ realizations, of the measurements estimation using the DMK algorithm constructed based on the noisy measurements, colored in blue, the DMK algorithm with the ``clean dynamics'', colored in cyan, and the particle filter, colored in gray.}
\label{fig:Toy1convergence1}
\end{figure}

Figure \ref{fig:Toy1convergence2} presents a comparison of the asymptotic RMSE values of the cleaned measurements obtained by the DMK algorithm (in blue), by the two competing non-parametric algorithms, GP and KKF (in light gray and in black, respectively) and by the particle filter (in dark gray), for two different SNR values. Plots (a) and (c) were created using an SNR of $0.67$ and plots (b) and (d) were created using an SNR of $0.18$. This figure depicts that the convergence rates of all the non-parametric algorithms are similar in both measurement coordinates and for both SNR values.

\begin{figure*}
\centering
\subfloat[]{\includegraphics[width=0.45\textwidth]{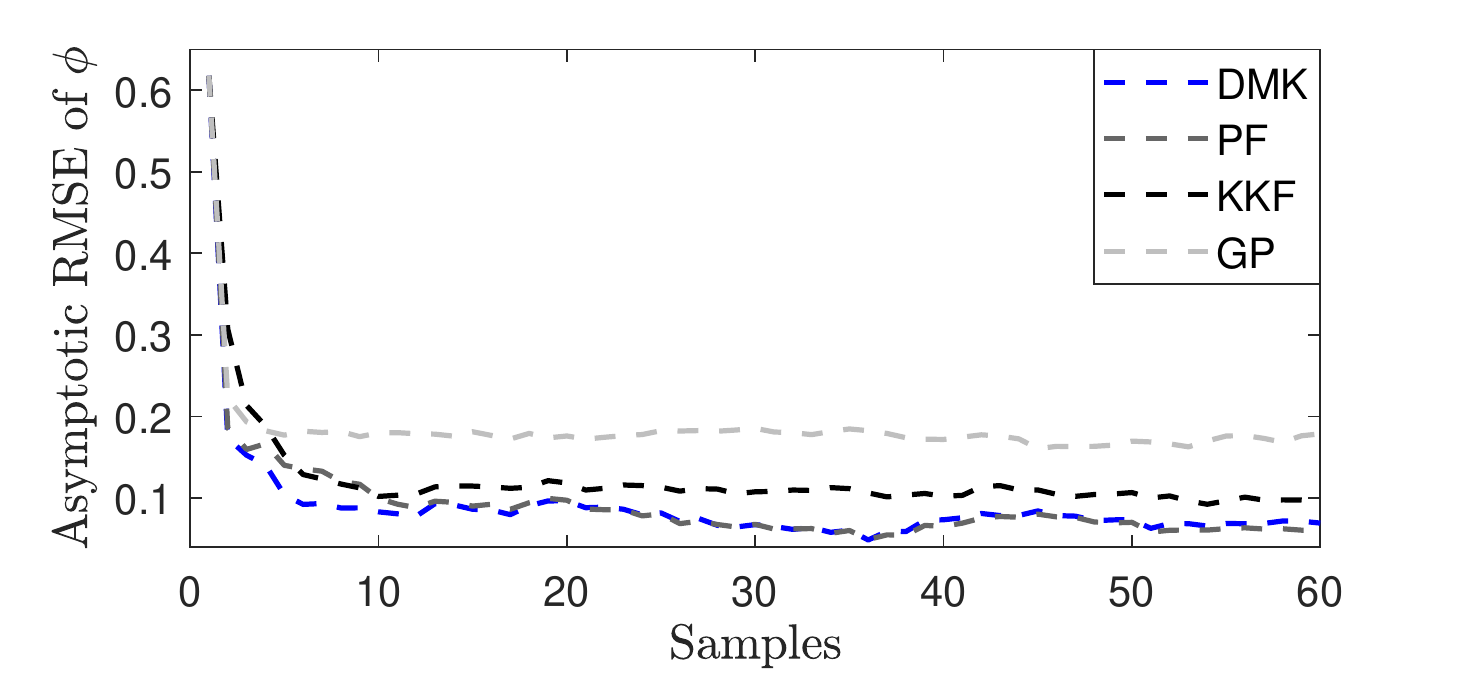}
}\subfloat[]{\includegraphics[width=0.45\textwidth]{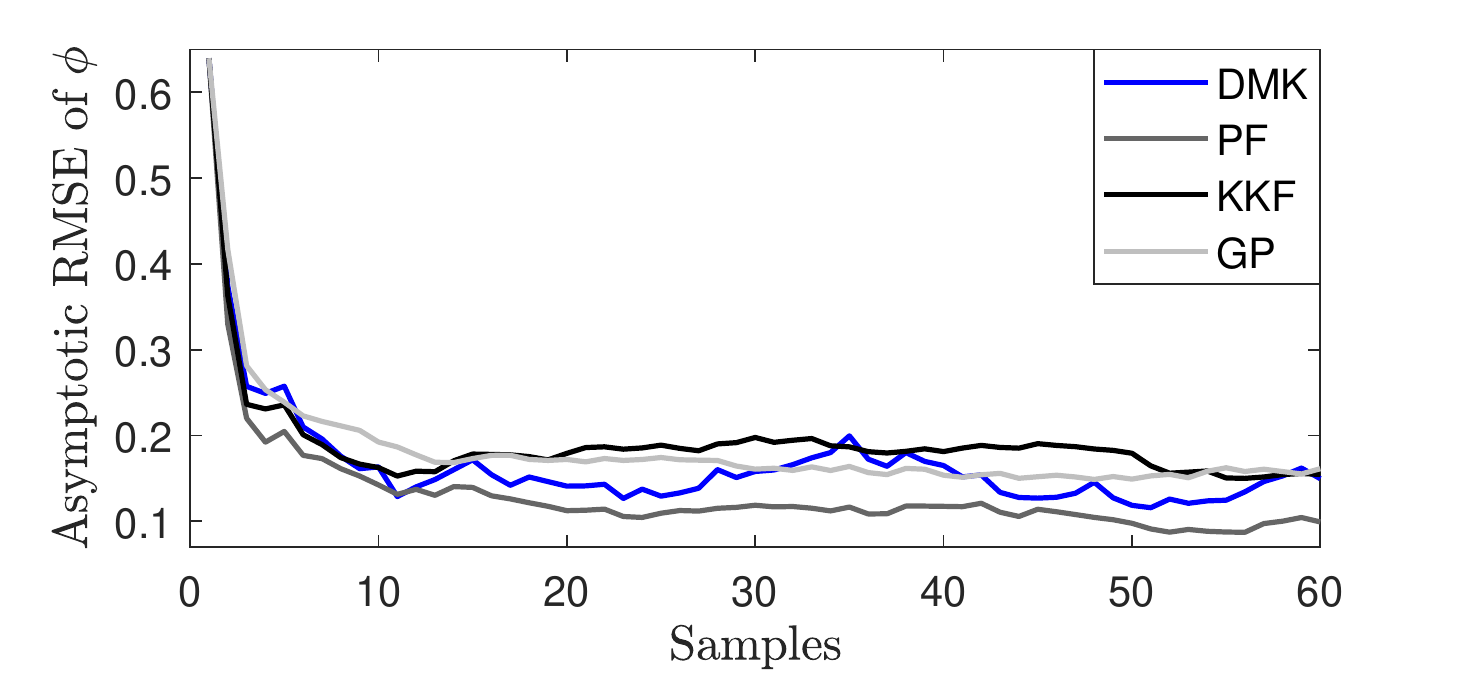}
}

\subfloat[]{\includegraphics[width=0.45\textwidth]{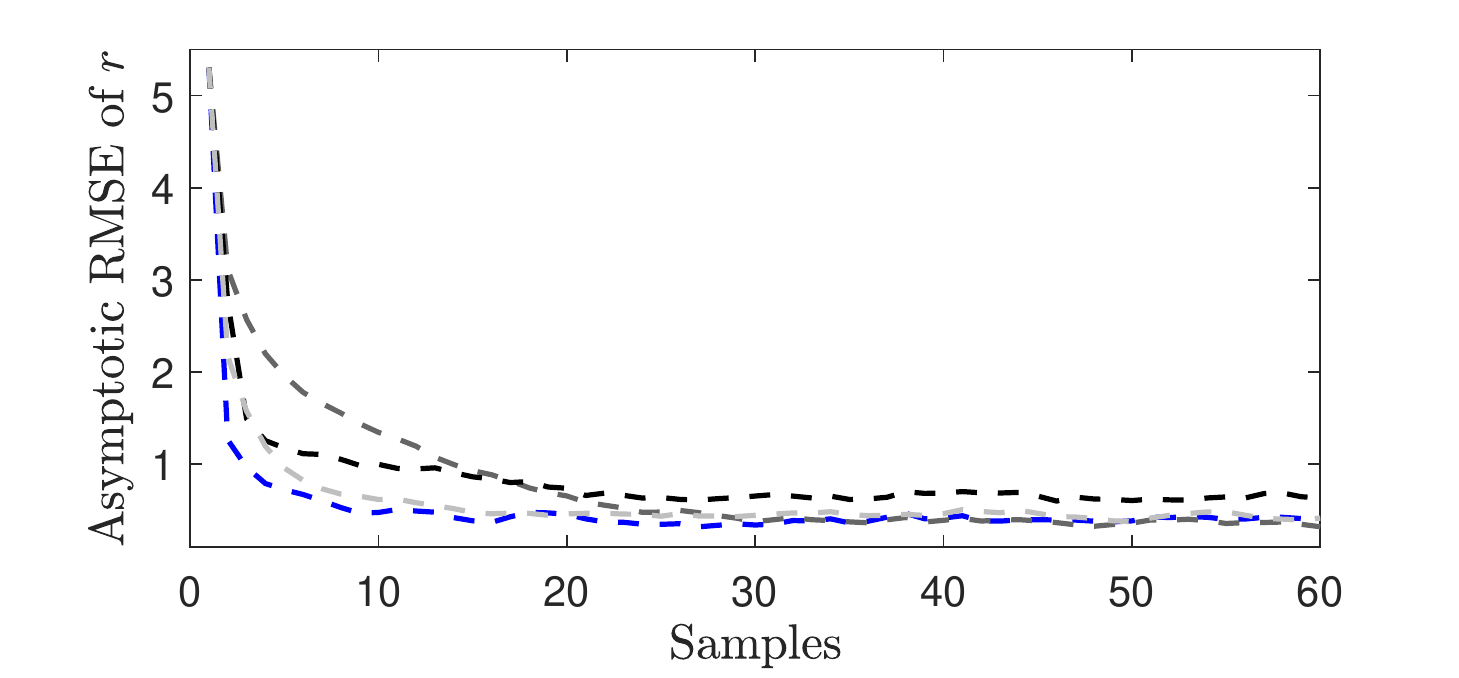}
}\subfloat[]{\includegraphics[width=0.45\textwidth]{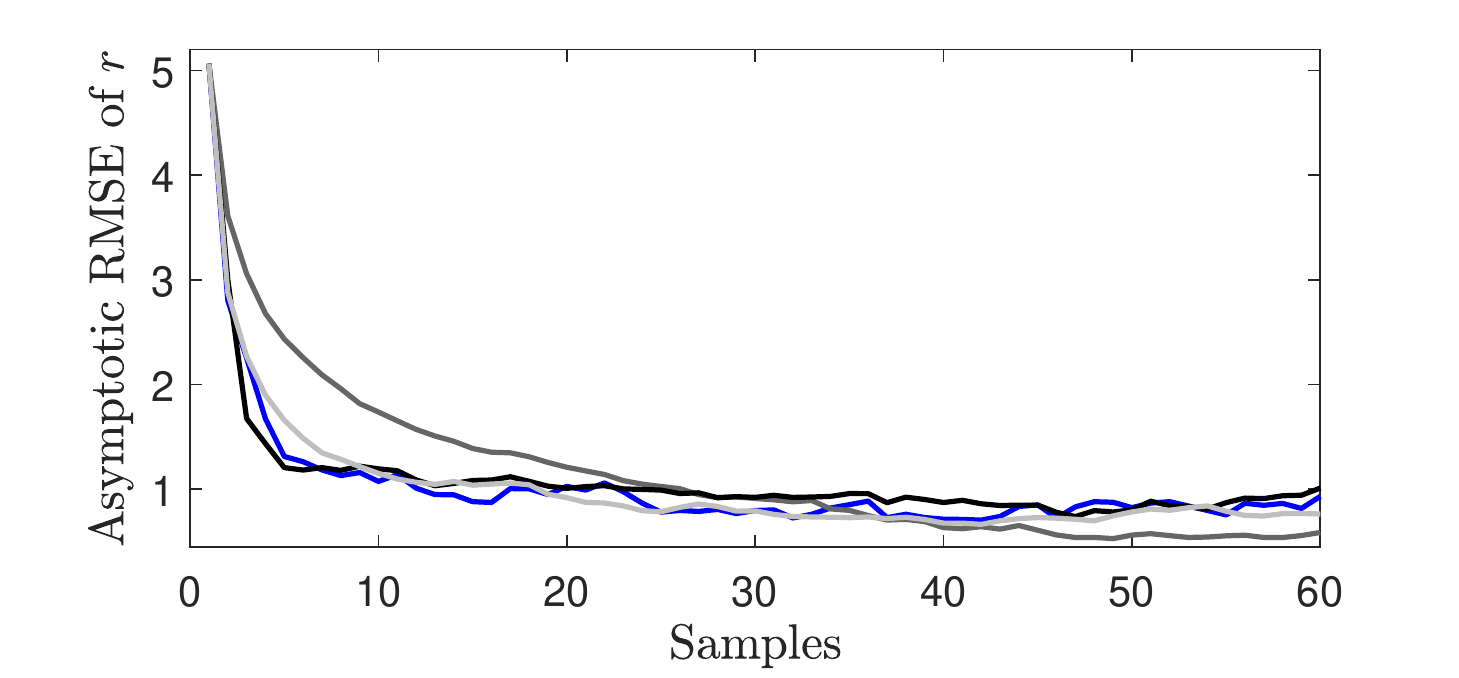}
}
\caption{Asymptotic RMSE, averaged over $50$ realizations, of the measurements estimation using the DMK (in blue), the PF (in dark gray), KKF, (in black) and GP (in light gray). The SNR was set to 0.67 in (a) and (c) and to 0.18 in (b) and (d). Plots (a) and (b) present the asymptotic RMSE values for measurement coordinate $\psi_n$ and plots (c) and (d) present the asymptotic RMSE values for measurement coordinate $r_n$.}
\label{fig:Toy1convergence2}
\end{figure*}

\subsection{Non-Gaussian Nonlinear Object Tracking\label{sub:Toy2}}
We now present an example which is based on the setting given in \cite{TalmonPNAS,shnitzer2016}. In this setting, the location of a radiating object moving on a 3D sphere is estimated based on measurements from three sensors, $\boldsymbol{s}^{(1)},\boldsymbol{s}^{(2)},\boldsymbol{s}^{(3)}$, which are modeled as ``Geiger Counters''. The movement of the object is defined by two underlying Langevin processes, describing the elevation and azimuth angles as follows
\begin{eqnarray}
\Delta\theta^{(1)}_{n+1} & = & \left(\frac{\pi}{2}\cdot c-c\cdot\theta^{(1)}_n\right)+u^{(1)}_n\\
\Delta\theta^{(2)}_{n+1} & = & \left(\frac{\pi}{5}\cdot c-c\cdot\theta^{(2)}_n\right)+u^{(2)}_n
\end{eqnarray}
where $b$ is the diffusion coefficient and $c$ is the drift rate parameter. 
In our experiments, $b$ was set to $0.01$ and $c$ was varied between $0.1$ and $1$ to simulate different trajectory properties.

The 3D location of the object at the $n$th time step is computed by:
\begin{eqnarray}
x^{(1)}_{n} & = & \cos\left(\theta^{(2)}_{n}\right)\sin\left(\theta^{(1)}_{n}\right)\nonumber\\
x^{(2)}_{n} & = & \sin\left(\theta^{(2)}_{n}\right)\sin\left(\theta^{(1)}_{n}\right)\label{eq:Toy2data}\\
x^{(3)}_{n} & = & \cos\left(\theta^{(1)}_{n}\right)\nonumber
\end{eqnarray}
We mark the 3D position of the object by $\boldsymbol{x}_{n}=\left[x^{(1)}_{n},x^{(2)}_{n},x^{(3)}_{n}\right]$.

The system measurements are given by three Poisson processes, with a rate parameter that is based on the 3D location of the object.
\begin{equation}
y^{(j)}_n\sim Pois\left(r^{(j)}_n\right)\,\,\,j=1,2,3
\end{equation}
where $r^{(j)}_n=\exp\left(-\left\Vert \boldsymbol{s}^{(j)}-\boldsymbol{x}_n\right\Vert\right)$. 

Finally, a Poisson noise process with a fixed rate parameter, $v^{(j)}_n\sim Pois\left(\lambda_{v}\right)$, is added to each sensor
\begin{equation}
z^{(j)}_{n} = y^{(j)}_n+v^{(j)}_n\label{eq:Toy2meas}
\end{equation}
where $z^{(j)}_{n}$ are the accessible system measurements.

Note that the presented setting is non-Gaussian and therefore the Kalman filter assumptions are not held. However, we show that DMK still provides good results, due to the incoporation of time dependencies, and is either better or comparable to the observer framework presented in \cite{shnitzer2016}.

In order to obtain an estimated representation of the system state (the azimuth and elevation angles) from the noisy measurements, $\left\{\boldsymbol{z}_n\right\}_{n=1}^{N}$, we apply DMK.

We simulated $300,000$ time samples of the two underlying angles $\theta^{(1)}_n,\theta^{(2)}_n$ with $\Delta t=0.5$ and constructed the measurements according to \eqref{eq:Toy2data} and \eqref{eq:Toy2meas}.
After obtaining the system measurements $\boldsymbol{z}_n$, we first performed a pre-processing stage, similarly to \cite{TalmonPNAS,shnitzer2016}. This includes constructing histograms for overlapping frames of $60$ time-samples of $\boldsymbol{z}_n$ and then calculating the modified Mahalanobis distance \eqref{eq:Mahalanobis} between pairs of histograms. 
The covariance matrices for the modified Mahalanobis distance between the measurements were calculated similarly to Subsection \ref{sub:Toy1}, based on the empirical covariance of overlapping windows of $20$ time frames, centered at each measurement.
At this point we are left with $300,000/60=5000$ system measurements.

The diffusion maps algorithm, described in Subsection \ref{sub:StateRecovery}, is applied to the measurements using the calculated Mahalanobis distance, with an empirical choice of $\epsilon$ as the median of the Euclidean distances. We obtain a set of eigenvectors and eigenvalues representing intrinsic properties of the system. However, these eigenvectors do not necessarily correspond to the true system state, $\boldsymbol{\theta}_n$, and can represent some linear combination of the state coordinates \cite{dsilva2015parsimonious}. Therefore, for evaluation purposes, we perform a linear regression on 3 eigenvectors, corresponding to the largest eigenvalues, and the true system state, based on $100$ samples.
Note that in this specific example, the linear regression provides a good representation of the underlying state using only a few diffusion maps coordinates. This is due to the choice of the underlying state equation and the lift function which places the data on a sphere. This system was specifically chosen in this application, since the measurements (spike trains) cannot be easily interpretable, in contrast to the example in Subsection \ref{sub:Toy1}. 
In Subsection \ref{sub:real_application}, a more complicated system is presented, where real neuronal spiking activity is analyzed. For such data, the underlying state equation and system model are completely unknown, and indeed, there, a larger number of coordinates is required to obtain a good representation of the true underlying state (the animal position).

Based on the resulting eigenvectors and eigenvalues we construct the Kalman filter, described in Subsection \ref{sub:Kalman}. The covariance matrices of the measurement noise, $R_n$, and the state noise, $Q_n$, were estimated from the data (according to the variance of the histogram measurements and the covariance of the obtained eigenvectors and eigenvalues) as described in Subsection \ref{sub:Toy1}. 

We compare our results to the observer framework, described in \cite{shnitzer2016}, with a choice of $\gamma=0.01$ (which led to the best results in this case).
In Figure \ref{fig:PNAS_ScatterColored}, a comparison between the DMK, the observer framework and the diffusion maps coordinates (without additional analysis) is presented. This figure contains $6$ identical scatter plots, each presenting the true underlying angles, $\theta^{(1)}_n,\theta^{(2)}_n$. Each plot is colored according to a different coordinate, plots (a) and (d) are colored according to the first and second estimated coordinates of the observer framework, plots (b) and (e) are colored according to the DMK estimation and plots (c) and (f) are colored according to the diffusion maps coordinates (DM). The color gradients in Figure \ref{fig:PNAS_ScatterColored} depict that the DMK significantly improves the estimation of the two underlying angles, compared with the diffusion maps coordinates and the observer framework. Moreover, the coordinates obtained by the observer framework suffer from inaccuracies at the boundaries of the data. This is visible for example, in plot (a), when $\theta^{(1)}_n<0.5$ and is due to the inaccuracy of the linear lift function at the boundaries. These inaccuracies are not apparent in the DMK coordinates which recover the true underlying angles more accurately even at the boundaries of the data.

\begin{figure*}
\centering

\subfloat[]{\protect\includegraphics[width=0.22\textwidth]{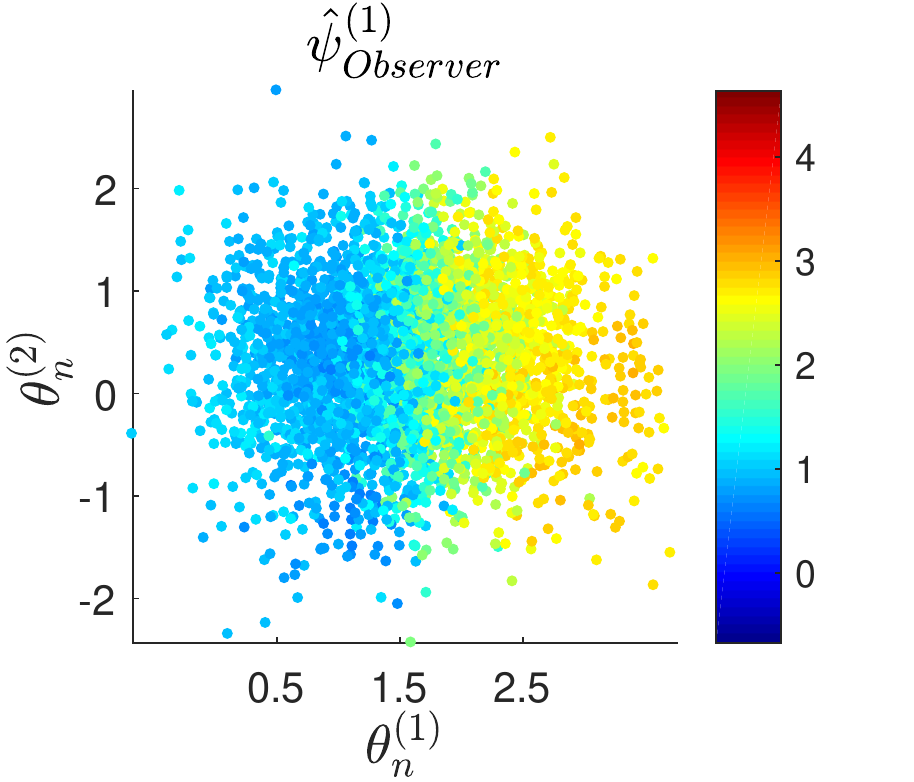}

}\subfloat[]{\protect\includegraphics[width=0.22\textwidth]{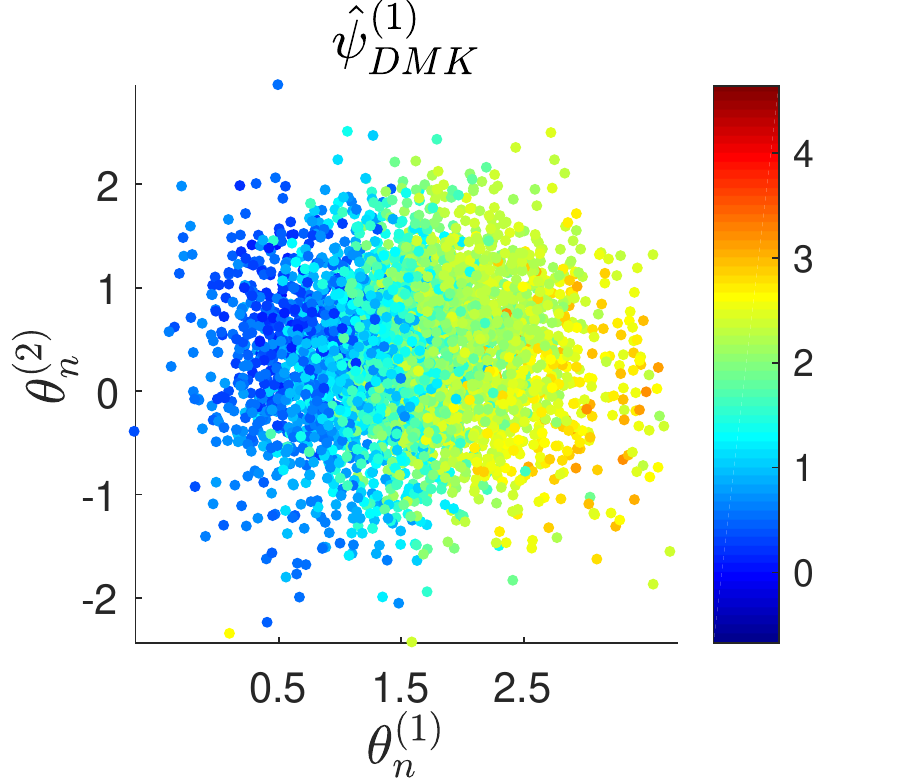}

}\subfloat[]{\protect\includegraphics[width=0.22\textwidth]{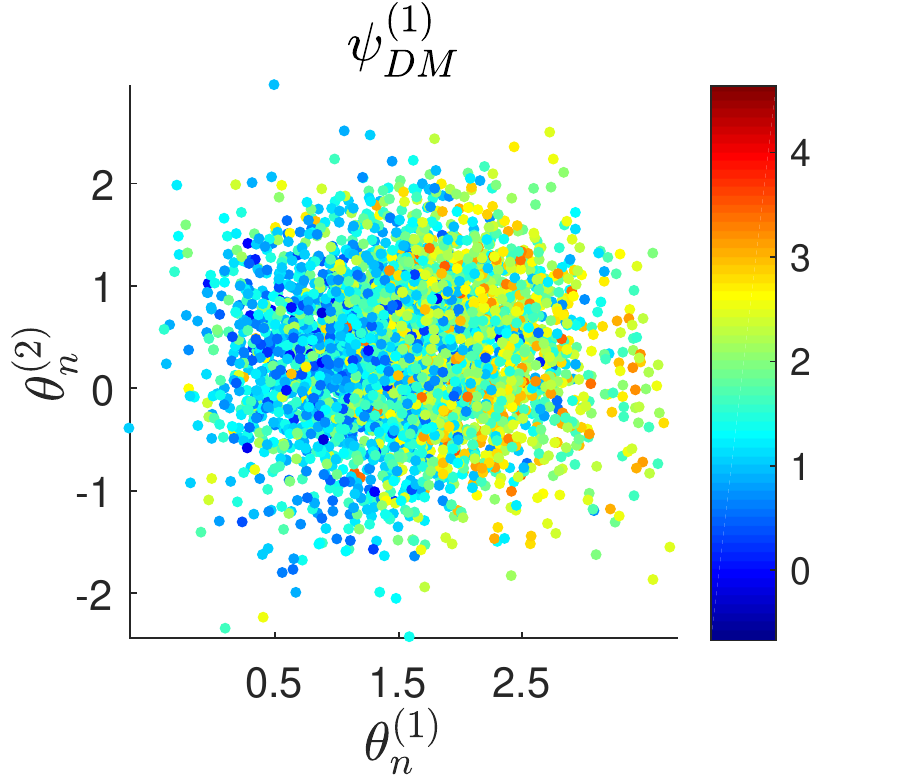} }

\subfloat[]{\protect\includegraphics[width=0.22\textwidth]{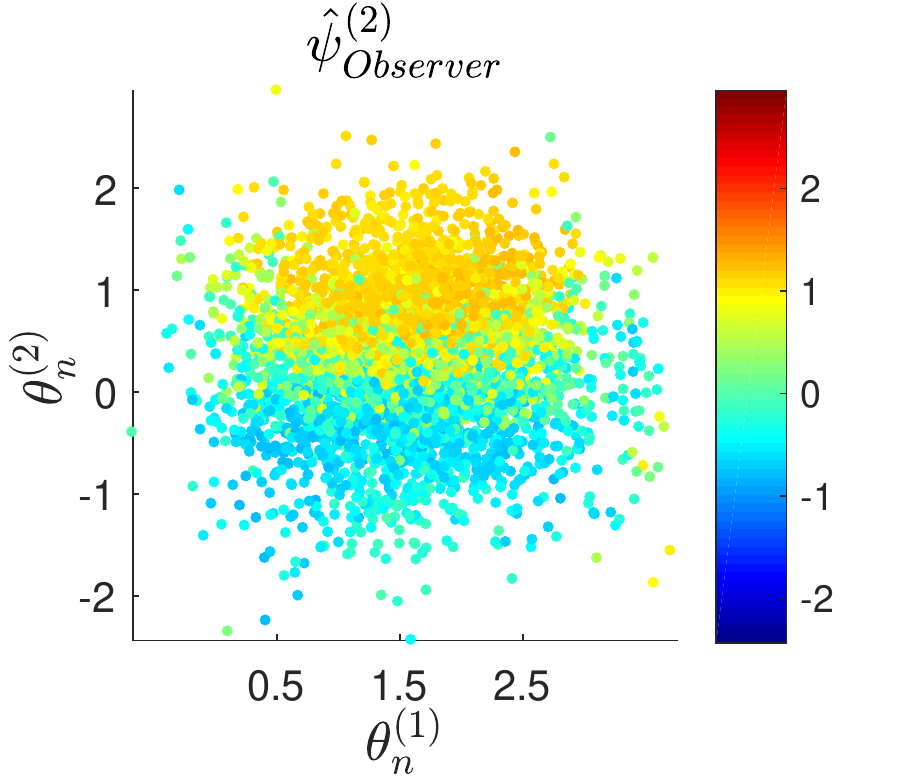}

}\subfloat[]{\protect\includegraphics[width=0.22\textwidth]{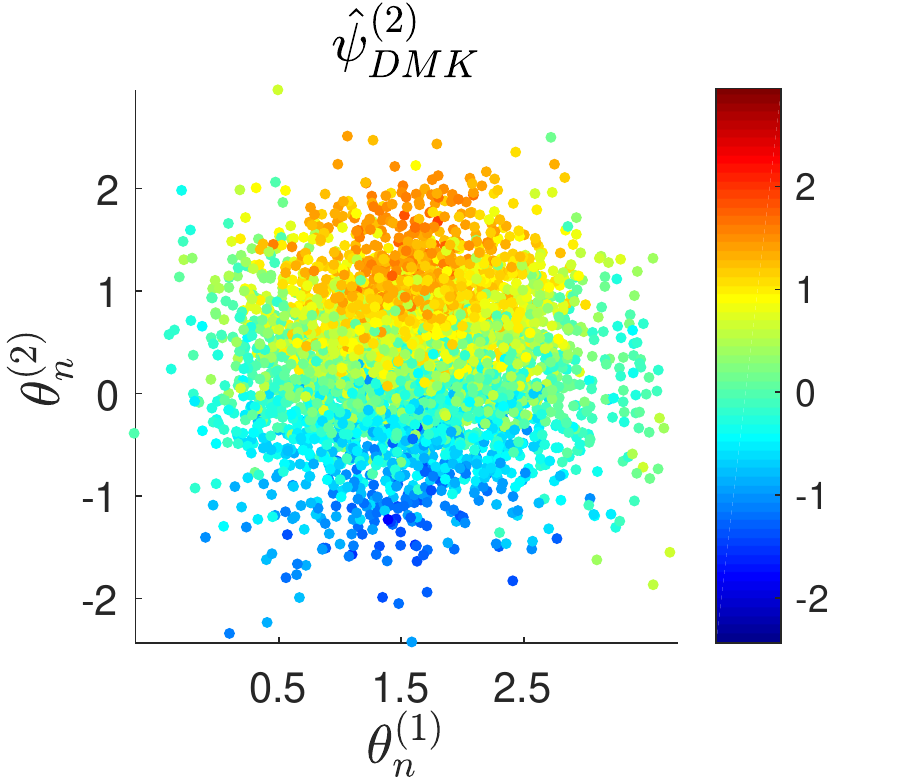}

}\subfloat[]{\protect\includegraphics[width=0.22\textwidth]{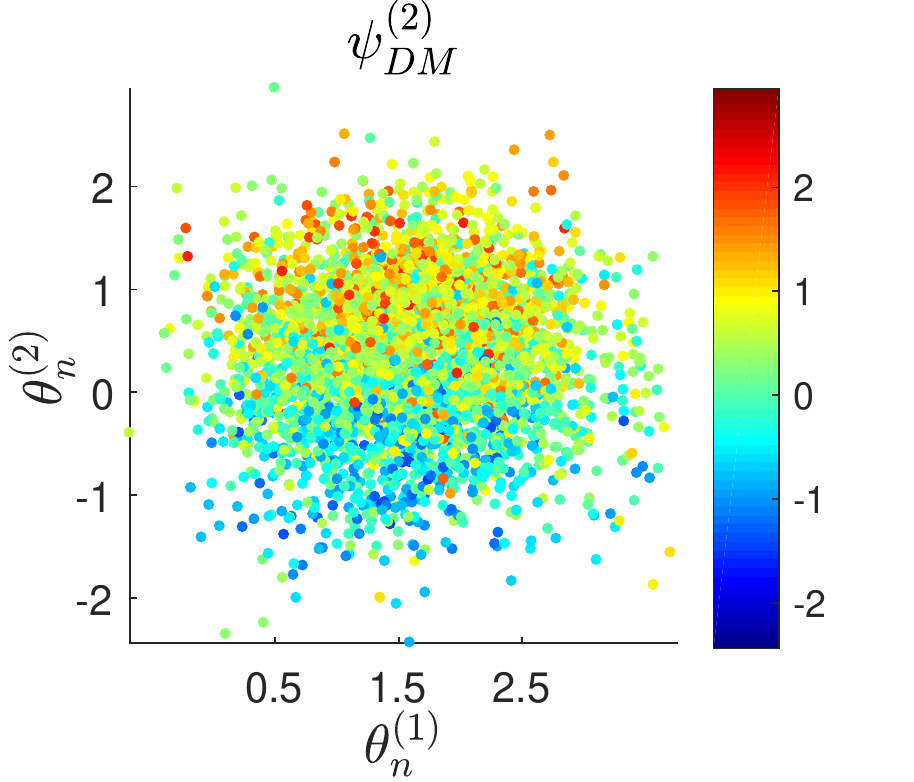} }

\caption{Scatter plots of the azimuth and elevation angles, colored by the coordinate estimates. The plots are colored according to the the first and second state estimates of the observer framwork (in (a) and (d)), the first and second state estimates of our suggested DMK filter (in (b) and (e)) and the first and second coordinates obtained by diffusion maps (in (c) and (f)).}
\label{fig:PNAS_ScatterColored} 
\end{figure*}

Figure \ref{fig:Toy2_correlation} presents the correlation between the DMK coordinates (after linear regression based on $100$ samples) and the true underlying states, $\theta^{(1)}$ and $\theta^{(2)}$, colored in blue. The average and standard deviation over $50$ realizations are presented. For comparison, we present the correlations of the observer coordinates, in green, and the diffusion maps coordinates, in black.
These plots depict that even though the Kalman model is inaccurate in this example, the DMK obtains result which are either better or comparable to the observer, in different drift-diffusion rate ratios.

\begin{figure}
\centering

\subfloat[]{\protect\includegraphics[width=0.25\textwidth]{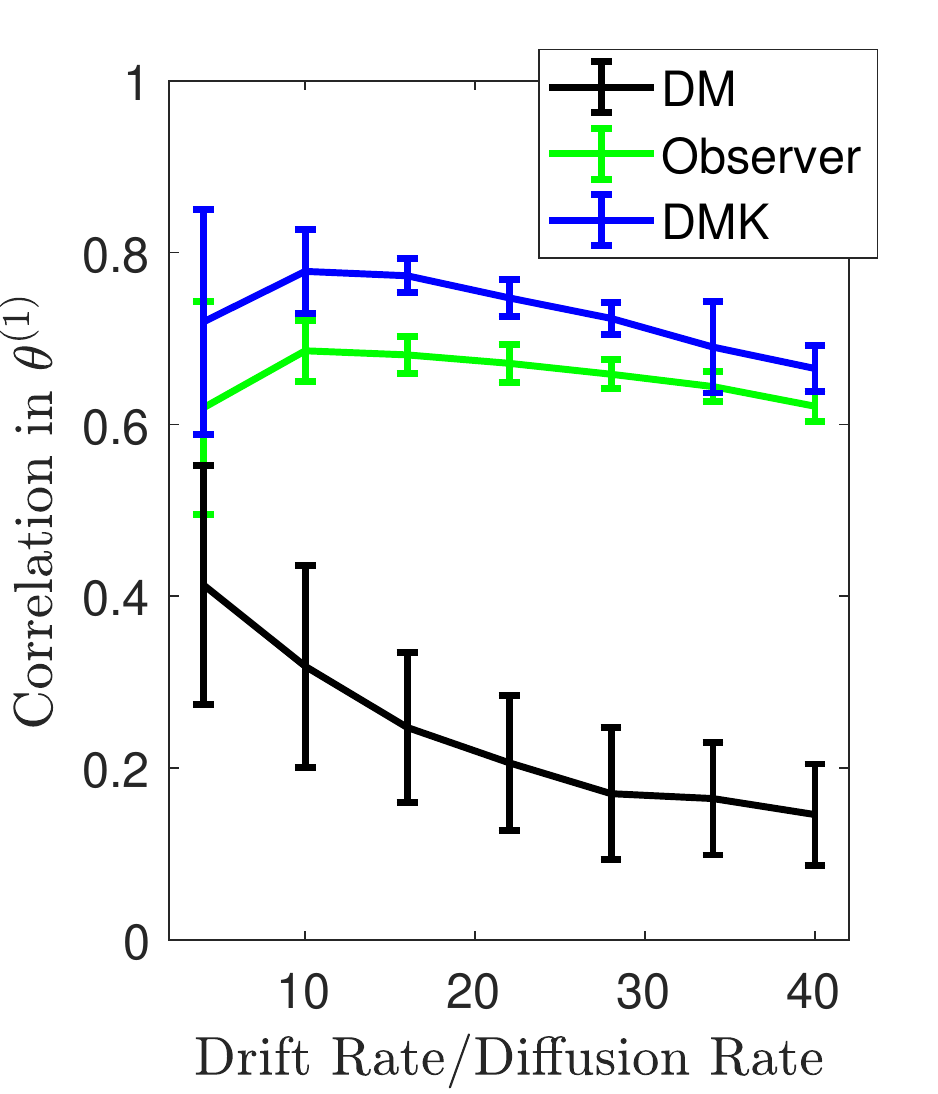}
}\subfloat[]{\protect\includegraphics[width=0.25\textwidth]{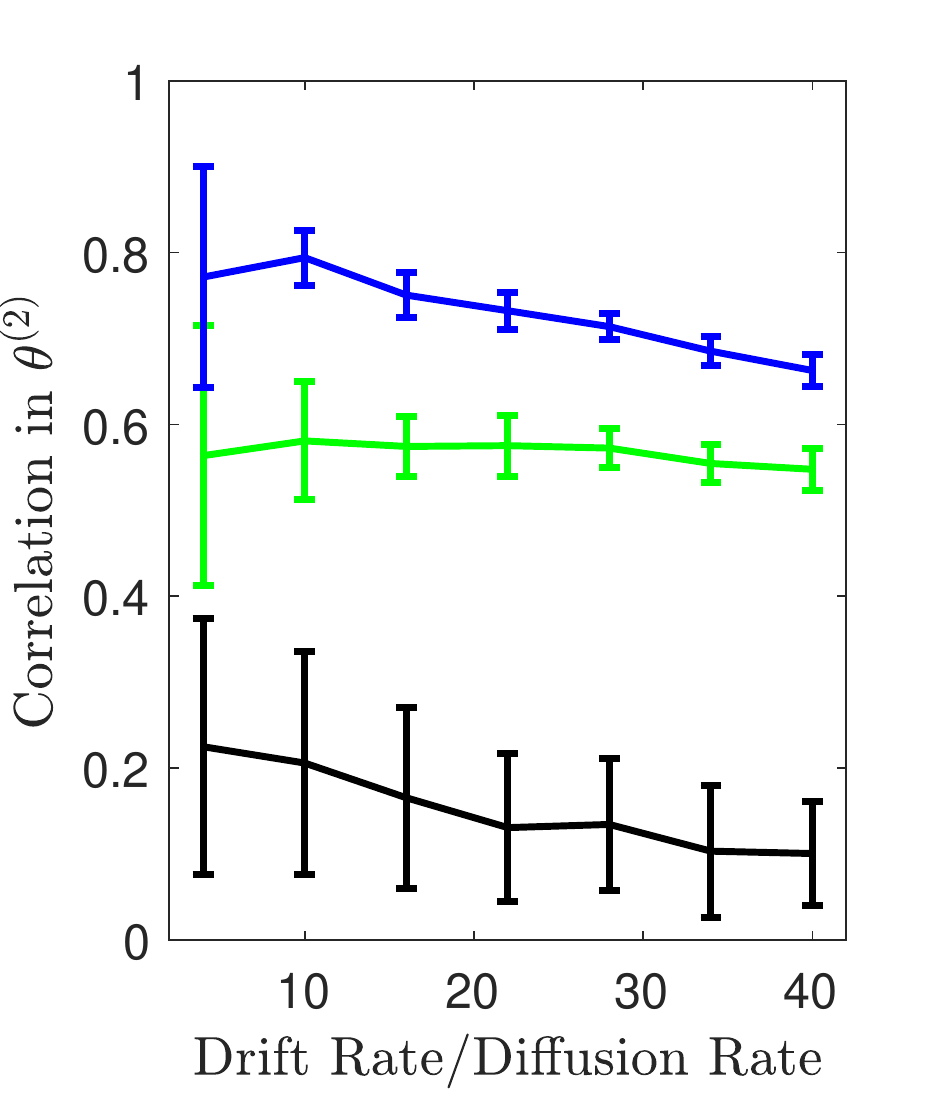}
}
\caption{Correlation between the DMK coordinates and the true underlying system state coordinates, in blue, compared with the observer coordinates, in green, and the diffusion maps coordinates, in black.}
\label{fig:Toy2_correlation} 
\end{figure}

\subsection{Location Tracking based on Rat Hippocampal Neuronal Activity\label{sub:real_application}}
To demonstrate our framework on real data, we consider a publicly available data set \cite{karlsson2015simultaneous}. 
This data set contains simultaneous recordings of hippocampal regions $CA1$ and $CA3$ or regions $CA1$ and $MEC$ in $9$ rats during a spatial alternation task in a W-shaped maze.
We show that when applying DMK to the neuronal spiking activity, it reveals a meaningful representation which is related to the 2D position of the rat. 

A related work \cite{wu2017gaussian} addressed this task using a \textit{parametric} framework that recovers a latent state governing the spike rates. Their framework is based on a Poisson process that generates the spiking data and two Gaussian processes, which model the latent state dynamics and the mapping between the latent state and the firing rate.
This parametric method assumes a specific model and is especially suited for neuronal spike train data, whereas our method is non-parametric and can be applied to a larger class of problems.

From the available data set \cite{karlsson2015simultaneous}, we focus on experiments in which the animals were actively moving in the maze. 
In addition, since some recordings include only a small number of active neurons, we restrict our analysis to experiments which contain more than $20$ active neurons and animals which have more than $3$ such experiments (for evaluation of error statistics).
The remaining data include six animals (abbreviated Bon, Ten, Cor, Eig, Fra and Mil) with $10$ to $24$ experiments each. The length of each experiment ranged between $7$ to $15$ minutes.

We apply the proposed method to the recorded neuronal spiking data and construct a new data-driven representation. 
We show that this new representation is highly related to the true position of the animal. 

For this purpose, we first bin the data to obtain a histogram for each neuron depicting the spike rate over time.
The histogram bin sizes were chosen empirically, such that at least one neuron is active at any time frame (bin). 
The bin sizes ranged between $1-2.6$ seconds (chosen differently in different animals). The bin sizes, number of experiments per animal and number of active neurons used in our analysis are summarized in Table \ref{tab:RatData}.

\begin{table*}
\centering
\begin{tabular}{|| c | c | c | c | c | c | c ||}
\hline
& Bon & Ten & Cor & Eig & Fra & Mil\\[0.5ex]
\hline\hline
Histogram bin size (sec) & $1$ & $1$ & $2.6$ & $1.6$ & $1$ & $1$\\[0.5ex]
\hline
Number of experiments & $24$ & $14$ & $10$ & $12$ & $18$ & $12$\\[0.5ex]
\hline
Number of neurons per exp. & $33-67$ & $42-60$ & $20-26$ & $29-59$ & $40-69$ & $25-38$\\[0.5ex]
\hline
\end{tabular}
\caption{Hipocampal neuronal activity - analysis parameters\label{tab:RatData}}
\end{table*}

The diffusion maps algorithm was applied to the obtained histograms, which are treated as measurements in the construction of the affinity matrix described in \eqref{eq:discreteKern}.
The covariance matrices for the modified Mahalanobis distance were estimated based on overlapping windows of $15$ time frames (histogram bins), and the kernel scale, $\epsilon$, was empirically set to be $3$ times the median of the distances, to avoid outliers in the diffusion maps coordinates.
The affinity matrix was normalized according to \eqref{eq:discreteNormKern} and its first $20$ (non-trivial) eigenvectors, corresponding to the largest eigenvalues, were computed.
The DMK was then applied to these eigenvectors, resulting in a new representation of the data, $\boldsymbol{\hat{\Psi}_n}=[\psi_n^{(1)},\dots,\psi_n^{(20)}]$.
The DMK covariance matrices were estimated based on the data and the calculated eigenvectors, similarly to Subsection \ref{sub:Toy1} and Subsection \ref{sub:Toy2}, where $Q_n(k,k)=\mathrm{var}(\lambda_k\psi_k)$ and $R_n$ was set to be a diagonal matrix with the variance of the histograms (measurements) on its diagonal.

Each experiment in each animal was analyzed and evaluated separately in order to avoid batch effects.

An example for the resulting representation of one experiment is presented in Figure \ref{fig:TruePos_DMK_color}. This figure presents $6$ identical plots containing the true x-y positions of the animal (marked by filled circles), where each plot is colored by a different DMK coordinate, according to $\mathrm{ln}\vert\hat{\psi}_i\vert$. 
We note that in this example, only $6$ eigenvectors were used as input to the DMK algorithm, since these $6$ coordinates provided a good representation for the position of the animal. 
In other experiments, $6$ eigenvectors were not always sufficient and therefore, for consistency, $20$ eigenvectors were used in all experiments in the following analysis.

\begin{figure*}
\centering
\includegraphics[width=0.75\textwidth]{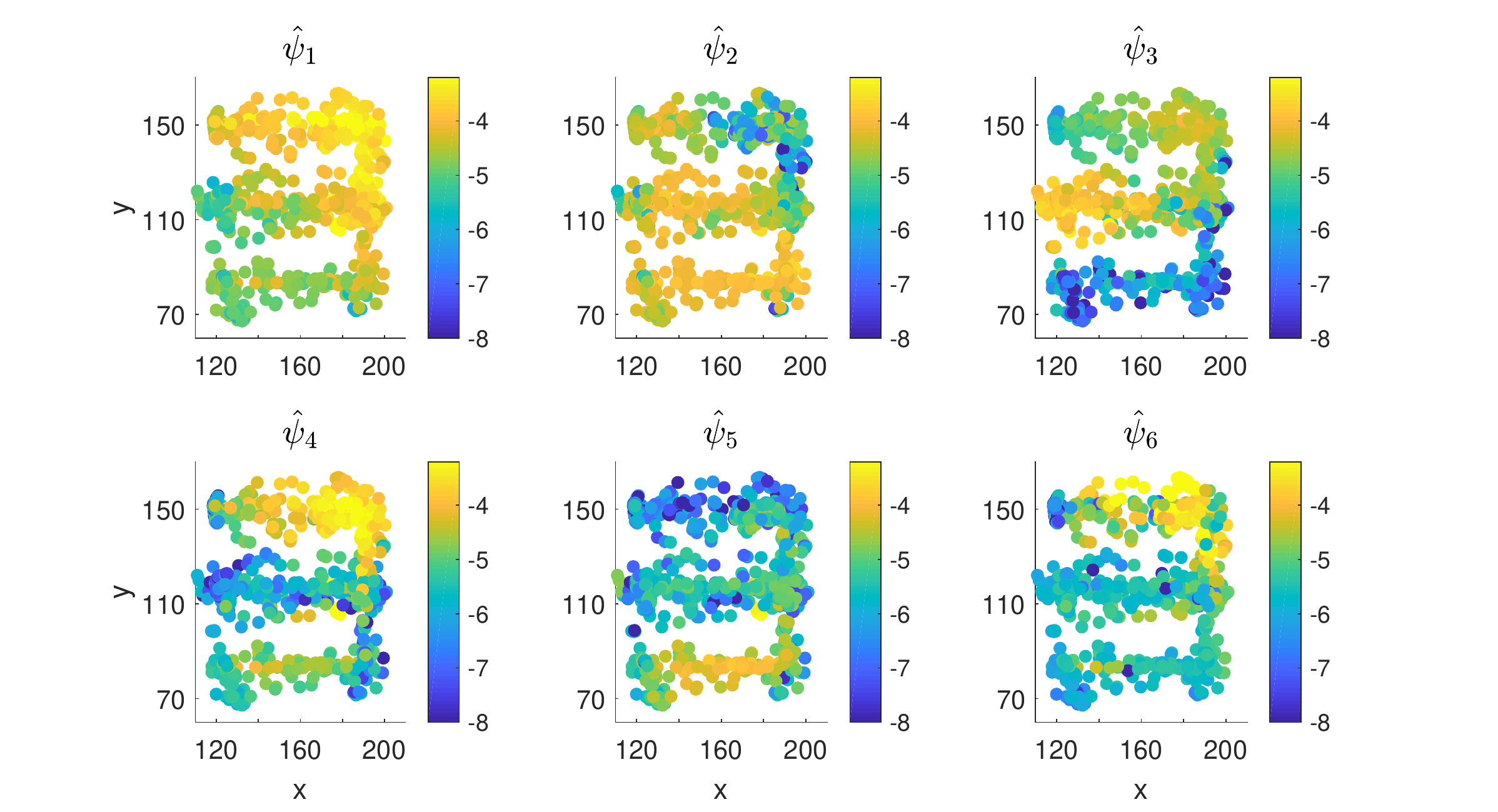}

\caption{True x-y positions of the animal in one experiment (filled circles), colored according to $6$ DMK coordinates ($ln\vert\hat{\psi}_i\vert$) calculated based on the spike rate histograms\label{fig:TruePos_DMK_color}}
\end{figure*}

Figure \ref{fig:TruePos_DMK_color} depicts that the DMK algorithm provides a meaningful representation for the animal location, since different 2D locations are highlighted in different coordinates (colored in yellow). 
By combining information from all coordinates, the 2D location can be inferred. For example, the middle arm in the W-shaped maze, $(x,y)\approx(160,110)$, is represented by high values (mostly) in coordinates $\hat{\psi}_2$ and $\hat{\psi}_3$, whereas the lower arm, $(x,y)\approx(160,80)$, is represented by high values (mostly) in coordinates $\hat{\psi}_2$ and $\hat{\psi}_5$.
In addition, Figure \ref{fig:TruePos_DMK_color} demonstrates that the obtained representation covers different regions in the 2D space in a relatively smooth manner. 

We emphasize that the DMK coordinates presented in Figure \ref{fig:TruePos_DMK_color} were obtained in a completely data-driven manner and with minimal model assumptions.

In order to quantitatively assess the quality of the new representation obtained by the DMK, we divided the data in each experiment into a train set, consisting of $80\%$ of the data, and a test set, consisting of $20\%$ of the data and performed cross-validation. 
We performed linear regression between the DMK coordinates representing the data in the train set and the true position of the animal and then used the regression parameters to estimate the animal's position based on the DMK representation of the test set.

An example for the resulting position estimation based on the test set of one experiment is presented in Figure \ref{fig:Position_trace}. 
This figure presents the estimated x position (top plot) and estimated y position (bottom plot) based on the DMK coordinates (in blue) compared with the true position (in black).
Figure \ref{fig:Position_trace} depicts that after the linear regression, the estimated position based on the DMK closely follows the true position of the animal.

\begin{figure}
\centering
\includegraphics[width=0.5\textwidth]{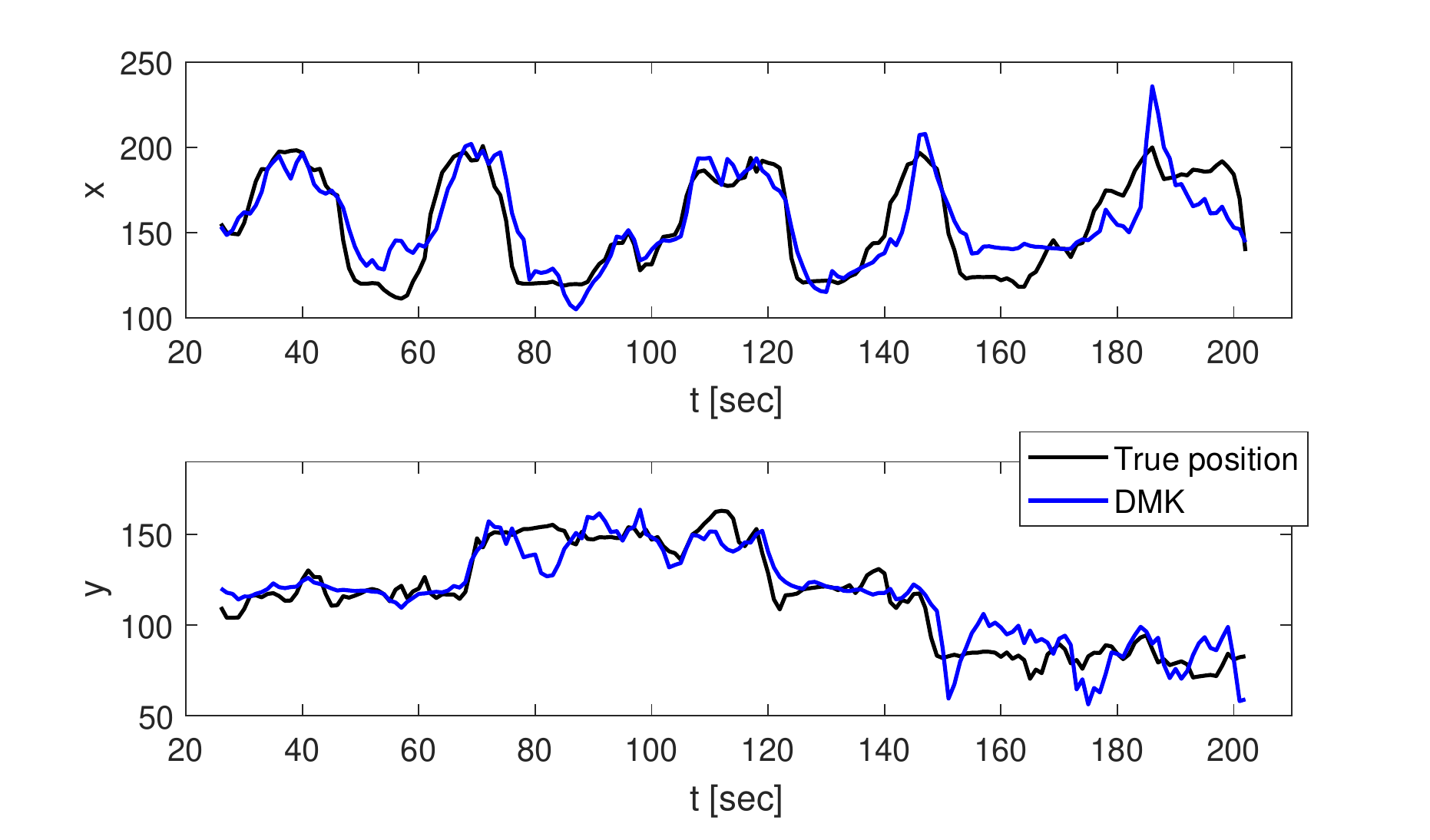}

\caption{Example of the estimated x-y positions of the animal in one experiment. The estimations are presented for the test-set, after applying linear regression to the DMK coordinates.\label{fig:Position_trace}}
\end{figure}

The corresponding 2D position of the animal in this experiment is presented in Figure \ref{fig:Position_2D}, where the position estimation based on the DMK coordinates is marked by blue 'x' and the true 2D position is marked by black circles. This plot depicts that after linear regression, the DMK coordinates are highly related to the 2D position of the animal in most locations. 
Note that the edge of the lower arm of the W-shaped maze is not represented properly by the DMK coordinates. This is consistent with the result in Figure \ref{fig:TruePos_DMK_color}, where none of the coordinates captures this specific location properly.

\begin{figure}
\centering
\includegraphics[width=0.4\textwidth]{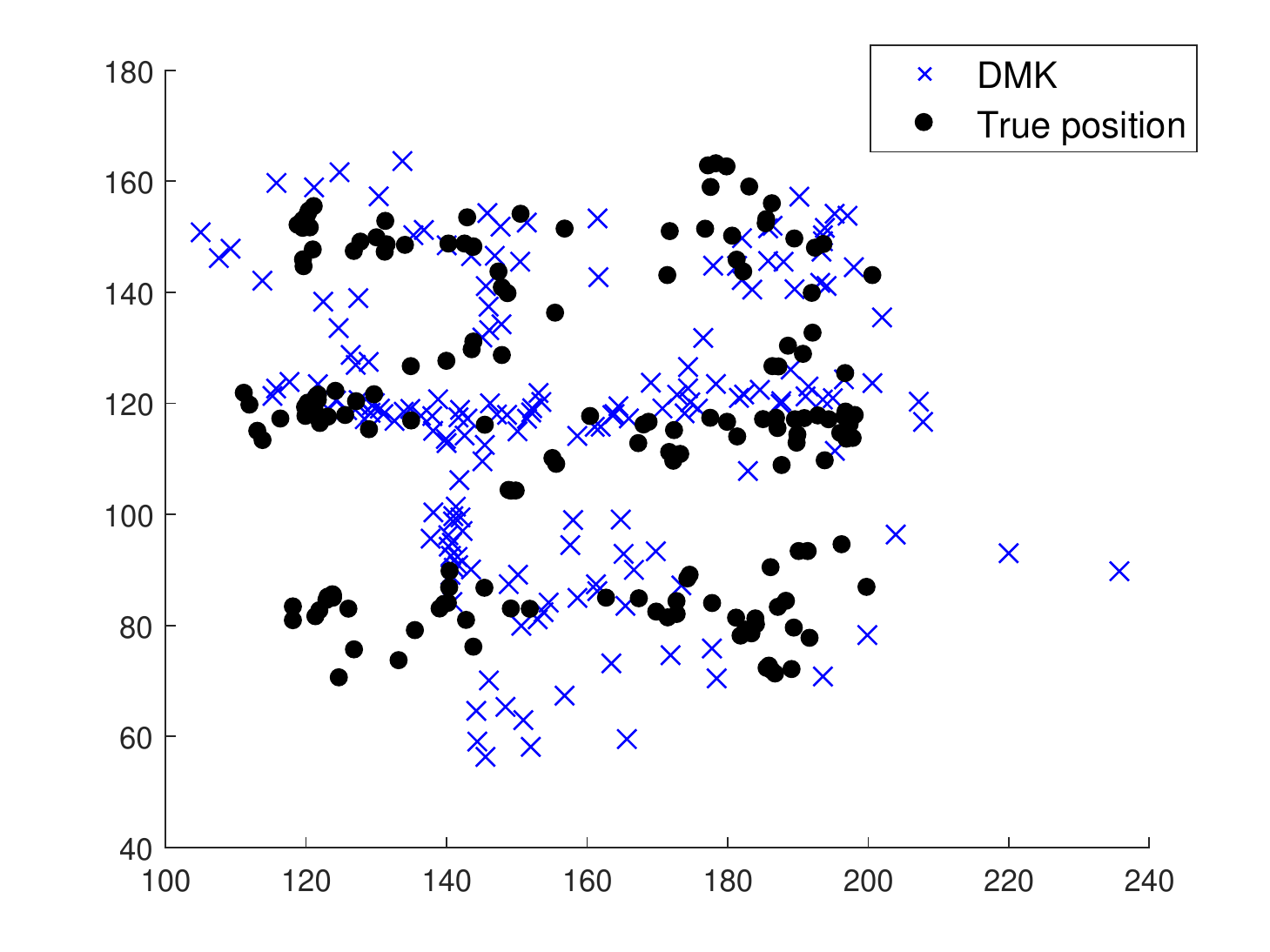}
\caption{Estimated 2D position in one experiment. The estimation is presented for the test-setm after applying linear regression to the DMK coordinates\label{fig:Position_2D}}
\end{figure}

For comparison, we applied linear regression to two additional data representations: (1) diffusion maps with no additional processing and (2) principal component analysis (PCA) applied to the spike rate histograms. 

We calculated the correlation between the estimated position and the true position for each experiment in all animals. 
Figure \ref{fig:Correlation_rat} presents the correlation values of the $x$ and $y$ positions, for each animal separately, based on the train set and on the test set. The average and standard deviation of the correlation values were calculated over 5-fold cross-validation, where the data was divided into $5$ consecutive segments.
This figure depicts that the DMK coordinates provide a meaningful representation that relates to the true position of the animal, since simply by linear regression, the DMK coordinates give rise to a good estimation of the true position. In addition, this representation is significantly better than the representations obtained by the diffusion maps coordinates and by PCA.

Note that the linear regression here was performed solely for a quantitative evaluation of the constructed coordinates. Our main result here is that the data driven DMK coordinates contain meaningful information regarding the location of the animal.

\begin{figure}
\centering
\includegraphics[width=0.5\textwidth]{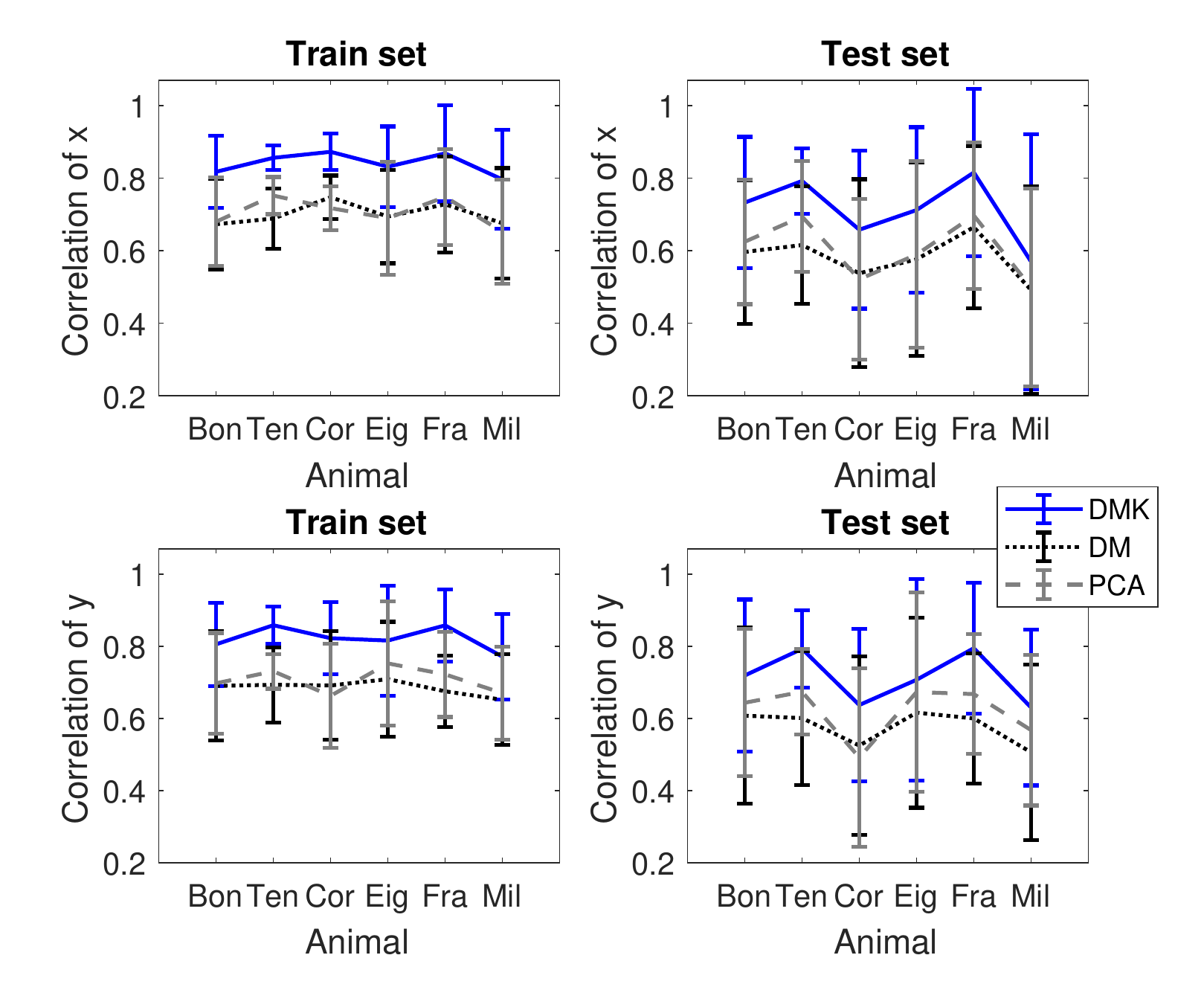}
\caption{Correlations between the estimated coordinates (after linear regression) and the true x-y positions, on the train-set and the test-set, for a 5-fold cross-validation.\label{fig:Correlation_rat}}
\end{figure}

We conclude by noting that in this example, the true system model is completely unknown and is not necessarily compatible with the assumed process form in \eqref{eq:Langevin}.
However, we demonstrated that our method can still be used to obtain a meaningful representation in this application.

\section{Conclusion\label{sec:Conclusion}}
In this work we addressed the analysis of high-dimensional nonlinear stochastic dynamical systems with measurement noise, and presented a non-parametric filtering framework in which a data-driven linear Kalman filter is constructed based on diffusion maps coordinates, utilizing their inherent dynamics and properties.
We showed that the presented framework recovers new coordinates that capture meaningful properties of the system given only a set of noisy measurements and with no further knowledge on the system properties. 
These coordinates may not be directly related to the underlying system state, yet they can be employed for filtering the measurements.
To obtain an estimate of a \emph{specific} underlying state, some alignment is required based on a training set, for example, using linear regression.

The theoretical justification of the proposed framework required few assumptions. 
First, the derivation of the recovered system dynamics was based on the assumption that the underlying system dynamics evolve according to gradient flows with constant diffusion.
Second, for compatibility of the recovered system with the Kalman filtering framework, we assumed that the leading diffusion maps coordinates are slowly evolving functions of the system state and are only mildly affected by the measurement noise. Several studies, e.g. \cite{Singer2009a,coifman2008diffusion} and \cite{dsilva2016data}, have shown that based on properties of the diffusion maps coordinates, this assumption commonly holds.
Third, we specified the conditions under which the devised system is observable and detectable, and showed that these conditions hold for an informed choice of the diffusion maps coordinates.
Although these assumptions are required for the theoretical derivations, the experimental results depict that our framework performs well in comparison to other non-parametric algorithms, even when these assumptions are not fully met.
We showed that the proposed method obtains an improved representation of dynamical systems compared with the diffusion maps coordinates and other non-parametric methods.

In the future, we plan to extend our method and address multi-modal data-sets arising from stochastic dynamical systems. We will devise methods for revealing the underlying common dynamics based on measurements from different sensors.

\section{Acknowledgements}
We would like to thank the Associate Editor and the anonymous reviewers for their helpful comments and suggestions.

\bibliographystyle{IEEEbib}
\bibliography{papers}

\end{document}